\documentclass[journal]{IEEEtran}

\usepackage{amssymb}
\usepackage{amsmath}
\usepackage{amsfonts}
\usepackage{multirow}
\usepackage{graphicx}
\usepackage{textcomp}
\usepackage{amsthm}
\renewcommand{\vec}[1]{\boldsymbol{#1}}
\usepackage{setspace}
\usepackage{subfigure}
\usepackage{algorithm}
\usepackage{algorithmic}
\usepackage[dvipsnames]{xcolor}
\usepackage{bm}
\usepackage{url}

\allowdisplaybreaks[4]

\ifCLASSINFOpdf
\else
\fi

\newtheorem{thm}{Theorem}
\newtheorem{lem}{Lemma}

\begin{document}

\title{A Double Auction for Charging Scheduling among Vehicles Using DAG-Blockchains}

\author{Jianxiong Guo,~\IEEEmembership{Member,~IEEE},
	Xingjian Ding,
	Weili Wu,~\IEEEmembership{Senior Member,~IEEE},
	and Ding-Zhu Du
	\thanks{Jianxiong Guo is with the Advanced Institute of Natural Sciences, Beijing Normal University, Zhuhai 519087, China, and also with the Guangdong Key Lab of AI and Multi-Modal Data Processing, BNU-HKBU United International College, Zhuhai 519087, China. (E-mail: jianxiongguo@bnu.edu.cn)
		
		Xingjian Ding is with the Faculty of Information Technology, Beijing University of Technology, Beijing 100124, China. (e-mail: dxj@bjut.edu.cn)
		
		Weili Wu and Ding-Zhu Du are with the Department of Computer Science, Erik Jonsson School of Engineering and Computer Science, The University of Texas at Dallas, Richardson, TX 75080, USA. (E-mail: weiliwu@utdallas.edu; dzdu@utdallas.edu)
		
		\textit{(Corresponding author: Xingjian Ding.)}
	}
	\thanks{Manuscript received xxxx; revised xxxx.}}

\markboth{Journal of \LaTeX\ Class Files,~Vol.~xx, No.~xx, June~2023}%
{Shell \MakeLowercase{\textit{et al.}}: Bare Demo of IEEEtran.cls for IEEE Journals}

\maketitle

\begin{abstract}
	Electric Vehicles (EVs) are becoming more and more popular in our daily life, which replaces traditional fuel vehicles to reduce carbon emissions and protect the environment. EVs need to be charged, but the number of charging piles in a Charging Station (CS) is limited and charging is usually more time-consuming than fueling. According to this scenario, we propose a secure and efficient charging scheduling system based on a Directed Acyclic Graph (DAG)-blockchain and double auction mechanism. In a smart area, it attempts to assign EVs to the available CSs in the light of their submitted charging requests and status information. First, we design a lightweight charging scheduling framework that integrates DAG-blockchain and modern cryptography technology to ensure security and scalability during performing scheduling and completing tradings. In this process, a constrained multi-item double auction problem is formulated because of the limited charging resources in a CS, which motivates EVs and CSs in this area to participate in the market based on their preferences and statuses. Due to this constraint, our problem is more complicated and harder to achieve truthfulness as well as system efficiency compared to the existing double auction model. To adapt to it, we propose two algorithms, namely Truthful Mechanism for Charging (TMC) and Efficient Mechanism for Charging (EMC), to determine an assignment between EVs and CSs and pricing strategies. Then, both theoretical analysis and numerical simulations show the correctness and effectiveness of our proposed algorithms.
\end{abstract}


\IEEEpeerreviewmaketitle

\section{Introduction}
\IEEEPARstart{T}{o} deal with the fossil energy crisis and reduce the emission of greenhouse gases, Electric Vehicles (EVs) have attracted more and more people's attention because of their great potential. Renewable energy will become the mainstream way of energy supply in the near future. Compared to traditional fuel vehicles, EVs have a number of advantages such as cost reduction, renewability, and environmental protection. With the development of EVs, a large number of Charging Stations (CSs) will be deployed in cities, which is different from current gas stations. Because of the tedious storage and transportation of gasoline, the deployment of gas stations is usually centralized. There is no such problem with electricity, especially after the rise of distributed energy. It can be seen that the deployment of CSs is more distributed, the number of CSs is larger, and the single CS is smaller. Besides, it usually takes more than half an hour to charge an EV, which is not the same thing as being able to refuel in an instant. This has also created a degree of administrative hardship.

In this paper, we consider such a scenario: In a smart area, there is a manager that is responsible for the  overall scheduling of EVs and CSs in this area. EVs hope to complete charging at the fastest speed and at the least cost, while CSs hope to maximize their profits by providing charging services. However, there are still some challenges that need to be resolved.

On the one hand, it lacks an effective approach to ensure the security of charging trading between EVs and CSs. Traditional centralized trading platforms depend on a trusted third party to manage and store every transaction between EVs and CSs. These platforms are plagued by some attacks \cite{xiong2021multi} \cite{wang2023incentive} such as single point of failure, denial of service, and privacy leakage \cite{zhang2019security}. A lot of existing research about charging management neglects the security and privacy protection of trading. The advent of blockchain technology has made it possible to improve these issues. Blockchain is a kind of decentralized ledger, which can realize security and privacy protection through the knowledge of modern cryptography and distributed consensus protocol. Based on that, we propose a Blockchain-based Charging Scheduling (BCS) system to manage the charging assignments between EVs and CSs in a secure and efficient manner. 

First, we are able to protect the contents of communication among the entities in this system from being tampered with or leaked by using digital encryption techniques. However, due to its confirmation delay, limited scalability, and the inherent uncertainty of consensus mechanisms \cite{guo2020blockchain} \cite{guo2022architecture}, traditional chain-based blockchains are not applicable to transactions with high-frequency characteristics. In addition, the incentive to miners, transaction fees, and restricted block sizes are also important factors that limit the use of such blockchains. To overcome the above drawbacks, the infrastructure of our BCS system is built on a Directed Acyclic Graph (DAG)-based blockchain \cite{popov2016tangle}. It has an asynchronous consensus process, which can make full use of network bandwidth to improve scalability. The DAG-based blockchain has no miners, thus it reduces latency and transaction fees, which makes frequent energy transactions between EVs and CSs possible \cite{popov2016tangle}. At the same time, it can guarantee the same security and decentralization as the traditional blockchain. In our BCS system, EVs work as light nodes, while CSs work as full nodes which are responsible for issuing new transactions, storing, and maintaining the whole blockchain system.

On the other hand, we have mentioned that charging is more time-consuming than fueling. For each CS in this system, the number of charging piles is limited. We can imagine that an EV goes to a CS for charging, but there is no idle charging pile at this moment that can charge it. This EV has to wait for other EVs to finish charging or go to other CSs, which is frustrating. For each EV in this system, it is more inclined to those CSs that locate near it or provide better services to it. Because EVs and CSs in this system belong to different entities, they are driven by their own utilities. Then, a natural question is how to assign the EV that submits a charging request in this system to a CS with a limited number of charging piles. Based on that, we formulate a constrained multi-item double auction model, where EVs are buyers and CSs are sellers. 

This model can be considered as an instance of a single-round multi-item double auction process \cite{yang2011truthful}. For the multi-item double auction, Jin \textit{et al. }\cite{jin2015auction1} \cite{jin2015auction2} proposed a resource allocation problem in mobile cloud computing and designed a truthful resource auction mechanism for the resource trading between mobile users (buyers) and cloudlet (sellers). However, this is a one-to-one mapping from buyers to sellers, that is, at most one mobile user can be assigned to one cloudlet at a moment. This scheme cannot be applied to solve our problem, because the number of EVs that can be assigned to a CS is more than one but less than the number of available charging piles in this CS. In other words, this is a many-to-one mapping from buyers to sellers. Secondly, to maximize the utilities of CSs (sellers), the assignment and price determination depends not only on the unit bids of EVs, but also on their charging amounts. Due to the limited number of charging piles in each CS as well as the different preferences and charging amounts of each EV, our constrained multi-item double auction model is distinguished from any existing double auction mechanisms. It is more complex and harder to achieve truthfulness and system efficiency. In order to solve this challenge, we design a Truthful Mechanism for Charging (TMC) first to determine the winners and transaction prices. By theoretical analysis, it meets the requirements of individual rationality, budget balance, truthfulness, and computational efficiency. Because of achieving truthfulness, TMC sacrifices a part of system efficiency. To improve the system efficiency, we design an Efficient Mechanism for Charging (EMC) that increases the number of successful trades (winning buyers) significantly more than that in TMC. Nevertheless, EMC is not able to ensure truthfulness for the buyers in extreme cases. The DAG-based blockchain and built-in smart contract implemented by the double auction enable our BCS system to work distributed and securely together.

The contributions of this paper are summarized as follows:
\begin{itemize}
    \item We investigate the challenges of EVs charging at present, integrate DAG-blockchain with charging scheduling, and proposed a complete framework for  charging scheduling and digital trading.
    \item We model the assignment between EVs and CSs as a constrained multi-item double auction. This model is first proposed by us, and its objective and constraint are fundamentally different from existing auctions. To solve it, we design a TMC algorithm that can guarantee truthfulness.
    \item In order to meet the rapid response of scheduling systems, we design an EMC algorithm that is more efficient than TMC at the expense of truthfulness.
    \item We build a simulation environment for the BCS system, and verify our expected goals for the system and auction theory through intensive simulations.
\end{itemize}


\textbf{Organizations: }In Sec. \uppercase\expandafter{\romannumeral2}, we discuss the-state-of-art work. In Sec. \uppercase\expandafter{\romannumeral3}, we introduce the background of DAG-based blockchain. In Sec. \uppercase\expandafter{\romannumeral4}, we present our BCS system and charging scheduling framework elaborately. In Sec. \uppercase\expandafter{\romannumeral5}, constrained multi-item double auction is formulated. Then, two solving mechanisms TMC and EMC are shown in Sec.\uppercase\expandafter{\romannumeral6} and \uppercase\expandafter{\romannumeral7}. Finally, we evaluate our proposed algorithms by numerical simulations in Sec.\uppercase\expandafter{\romannumeral8} and show the conclusions in Sec. \uppercase\expandafter{\romannumeral9}.

\section{Related Work}
Recently, blockchain has been used as an effective method to deal with the issues of transactions generated by peer-to-peer (P2P) energy trading among EVs. Kang \textit{et al. }\cite{kang2017enabling} exploited a P2P electricity trading system with consortium blockchain to motivate EVs to discharge for balancing local electricity demand. Wu \textit{et al. }\cite{wu2017two} studied energy scheduling in office buildings through combining renewable energies and workplace EV charging, and used stochastic programming to manage the uncertainty of charging. Liu \textit{et al. }\cite{liu2018adaptive} proposed an EV participation charging scheme for a blockchain-enable system to minimize the fluctuation level of the grid network and charging cost of EVs. Su \textit{et al. }\cite{su2018secure} designed a contract-based energy blockchain in order to make EVs charge securely in the smart community, while they implemented a reputation-based delegated Byzantine fault tolerance consensus algorithm. Zhou \textit{et al. }\cite{zhou2019secure} developed a secure and efficient energy trading mechanism based on consortium blockchain for vehicle-to-grid that exploited the bidirectional transfer technology of EVs to reduce the demand-supply mismatch. Xia \textit{et al. }\cite{xia2020bayesian} proposed a vehicle-to-vehicle electricity scheme in blockchain-enabled Internet of Vehicles to address the driving endurance issue of EVs. Guo \textit{et al. }\cite{guo2021reliable} designed a lightweight blockchain-based information trading framework to realize a real-time traffic monitoring.

However, due to its confirmation delay, limited scalability, and lack of computational power, traditional chain-based blockchains cannot be deployed in the systems associated with EVs. The DAG-based blockchain (Tangle) \cite{popov2016tangle} emerged as a new type of blockchain, which reduces the reliance on computational power and improves the throughput significantly. Huang \textit{et al. }\cite{huang2019towards} presented a DAG-based blockchain with a credit-based consensus mechanism for power-constrained IoT devices. Hassija \textit{et al. }\cite{hassija2020blockchain} exploited the DAG-based blockchain to support the increasing number of transactions in the vehicle-to-grid network. In this paper, our BCS system is built on DAG-based blockchain as well because of its limited computational power and high-frequency characteristics.

Auction theory has been widely applied in many different fields, such as mobile crowdsensing \cite{yang2015incentive} \cite{lin2021multi} \cite{zhang2023incentive}, mobile cloud computing \cite{jin2015auction1} \cite{jin2015auction2} \cite{ding2020incentive}, and energy trading \cite{zhao2022comparisons} \cite{an2023distributed}. Here, we only focus on double auctions, where buyers (resp. sellers) submit their bids (resp. asks) to an auctioneer. For example, Borjigin \textit{et al. }\cite{borjigin2018broker} proposed a double auction method to maximize the profit of participants that can be applied in service function chain routing and NFV price adjustment. There are two classic models: McAfee double auction \cite{mcafee1992dominant} and Vickrey-based model \cite{parkes2001achieving}. They only considered homogeneous items and the Vickrey-based model cannot satisfy the truthfulness, unfortunately. To heterogeneous items (multi-item), Yang \textit{et al. }\cite{yang2011truthful} designed a truthful double auction scheme for the cooperative communications, where the auctioneer first uses an assignment algorithm based on its design goal to get candidates and mapping from buyer to seller, then apply McAfee double auction to determine the winners and corresponding transaction prices. 

According to \cite{yang2011truthful}, Jin \textit{et al. }\cite{jin2015auction1} \cite{jin2015auction2} proposed a truthful resource auction mechanism in mobile cloud computing, which is the most relevant to our paper and has been introduced in Sec. \uppercase\expandafter{\romannumeral1}. However, there are obvious differences in our work. In our constrained multi-item double auction model, the number of buyers that can be assigned to the same seller is constrained. Furthermore, we should not only consider the unit bids of buyers, but also consider their charging amounts. It is more difficult to achieve truthfulness and system efficiency, which is our main contribution to this paper.

\section{Background for DAG-Blockchains}
Blockchain is an emerging technology that acts as a decentralized ledger or database since Nakamoto published his original prototype in 2008 \cite{nakamoto2008bitcoin}, which is implemented by modern cryptographic technologies and distributed consensus algorithms. Because of these technical characteristics, it takes more than half of the computational power to tamper with transactions in the blockchain, which prevents attacks from malicious nodes effectively. In the process of communication and storage, All transactions are encrypted digitally, which realizes the anonymity and privacy protection of the blockchain system. They can enable users who do not trust each other to trade freely in a secure and reliable environment without a third-party platform. Then, a number of real applications flourished, such as Bitcoin \cite{nakamoto2008bitcoin}, Ethereum \cite{wood2014ethereum}, and Hyperledger fabric \cite{androulaki2018hyperledger}, all of which were based on the chain-based blockchain.

\begin{figure}[!t]
	\centering
	\subfigure[The chain-based blockchain]{
		\includegraphics[width=2.5in]{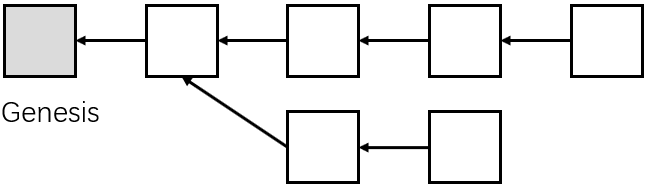}
	}%
	
	\subfigure[The DAG-based blockchain]{
		\includegraphics[width=2.5in]{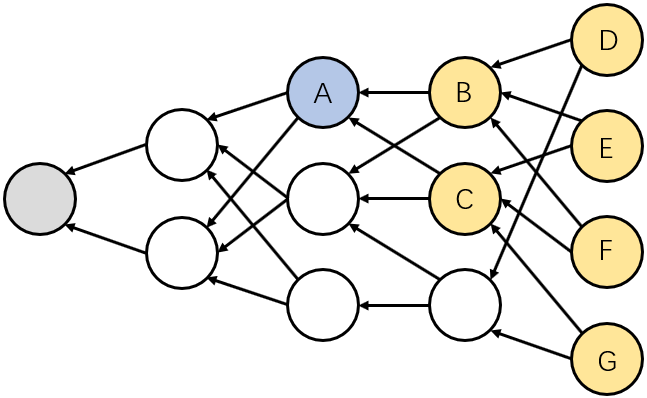}
	}%
	\centering
	\caption{The comparison of chain-based and DAG-based blockchian, where the square nodes in (a) are blocks and the circle nodes in (b) are transactions.}
	\label{fig1}
\end{figure}

The chain-based blockchain is shown in Fig. \ref{fig1} (a), where a block contains a certain number of transactions. All users maintain the longest main chain jointly, and only the transactions in the main chain can be considered legal. However, this chain-based structure is limited by its high requirement for computational power because of its consensus mechanism based on hashing puzzles, thereby it cannot be applied to devices with limited power. Secondly, it generates and verifies the blocks in a sequential and synchronous manner, resulting in unacceptable confirmation delay and scalability. Take Bitcoin as an example, its average throughput is about 7 \textit{TPS} (\textit{Transactions Per Second}) \cite{kogias2016enhancing}. Thus, it is not suitable for those real-time application scenarios with high-frequency characteristics. Here, the system we design is based on EVs which are both resource-limited and high-frequency, hence the chain-based structure is difficult to achieve our goal.

To improve the scalability, a new blockchain architecture based on the Directed Acyclic Graph (DAG) was proposed, called DAG-based blockchain or tangle \cite{popov2016tangle}. The DAG-based blockchain is shown in Fig. \ref{fig1} (b), where each node represents a transaction instead of a block. When a new transaction is issued, it must validate two tips which are previous transactions attached in the tange but not verified by any others. This validation is denoted by a directed edge in the tangle. This new transaction will be verified by the other upcoming transactions. The time required to verify a tip (confirmation delay) depends on the tip selection algorithms \cite{popov2016tangle} and the rate of new transactions.

Generally, there is a weight associated with each transaction, which is proportional to the difficulty of the hashing puzzle defined by itself. When issues a new transaction, it has to find a random number \textit{nouce} such that
\begin{equation}
\text{Hash}(\text{Transaction, Timestamp, nouce})\leq \text{Target}
\end{equation}
where the smaller the target implies the greater the weight. Until now, the newly issued transaction has been completed and waits to be verified by subsequent transactions. When can we consider a transaction valid? This is related to its cumulative weight. The cumulative weight of a transaction is the weighted sum of transactions that approve it directly or indirectly. Shown as Fig.\ref{fig1}(b), the cumulative weight of transaction $A$ equals the weighted sum of transactions $A$, $B$, $C$, $D$, $E$, $F$, and $G$. Consider a transaction, the larger cumulative weight can only be achieved by consuming more computational power, thereby it is more likely to be legal if it has a larger cumulative weight. A transaction is believed to be legal if and only if its cumulative weight exceeds a pre-defined threshold. Suppose that most power is controlled by legal users, we can distinguish those transactions issued by malicious users since the valid transactions will be verified by other legal users and their cumulative weight will be larger and larger.

Different from those synchronous consensus mechanisms in the chain-based blockchain, the consensus process in the DAG-based blockchain is completed in an asynchronous approach. This eliminates an inherent defect of the chain-based blockchain that has on consensus finality because of forking and pruning \cite{hassija2020distributed}. Besides, it is able to defend against possible attacks, such as a single point of failure, Sybil attack, lazy tips, and double-spending. According to that, the DAG-based blockchain can not only provide us with a secure and reliable trading environment but also improve the scalability and reduce the requirements on the computational power of devices. Then, a lot of real systems based on DAG-based blockchains emerged, such as IOTA \cite{popov2016tangle}, Byteball \cite{churyumov2016byteball}, and Nano \cite{lemahieu2018nano}. Finally, the devices in our charging scheduling system are resource-limited and trade with others frequently. Therefore, the DAG-based blockchain is an ideal choice to act as infrastructure to achieve security and privacy protection.

\section{Blockchain-based Charging Scheduling System}
In this section, we demonstrate the overview of our Blockchain-based Charging Scheduling (BCS) system by introducing entities and a charging scheduling framework.
\subsection{Entities for BCS System}
Consider a smart area, there are a certain number of CSs with charging piles available to charge the EVs in this area. Then, it exists a manager that is responsible for managing entities and executing the charging scheduling between EVs and CSs. Therefore, there are three main entities in this system shown as follows.

\begin{itemize}
	\item \textbf{Electric Vehicle (EV):} The EVs running in this area play the part of requesters. When it is low on power, the EV will request for charging services from the manager.
	\item \textbf{Charging Station (CS):} The CSs located in this area play the part of providers to charge the EVs. A CS has many charging piles, and each charging pile can charge one EV at a time. When it is idle, the CS will inform the manager of its status information.
	\item \textbf{Manager:} The manager works as an energy scheduler. Each EV sends a request about its charging preference and each CS submits its status information on how many charging piles are available to the manager. Then, the manager acts as an auctioneer to perform a constrained double auction mechanism between EVs and CSs that assigns EVs to CSs and determines the transaction prices. The price charged to EV and payment rewarded to CS are determined by the transaction prices.
\end{itemize}

Thus, the BCS system can be denoted by $\mathbb{B}=\{M,\mathbb{V},\mathbb{C}\}$ where $\mathbb{V}=\{V_1,V_2,\cdots\}$ is the set of EVs, $\mathbb{C}=\{C_1,C_2,\cdots\}$ is the set of CSs, and $M$ is manager. Then, we define a transaction $T_{ij}$ between $V_i\in\mathbb{V}$ and $C_j\in\mathbb{C}$ as the charging trading and digital payment record between them. This transaction is stored in the blockchain and includes pseudonyms of $V_i$ and $C_j$, data type, transaction details, and timestamps of generation. In order to ensure security and privacy protection, the transaction is encrypted by their digital signatures and the payment is made in the form of charging coins.

\begin{figure}[!t]
	\centering
	\includegraphics[width=2.5in]{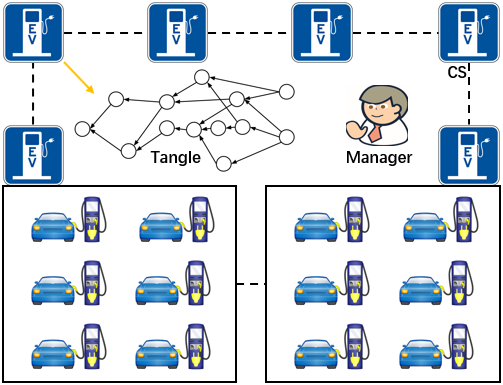}
	\centering
	\caption{The architecture of blockchain-based charging scheduling system in a smart area.}
	\label{fig2}
\end{figure}

\subsection{Architecture for BCS system}
The infrastructure of BCS system $\mathbb{B}=\{M,\mathbb{V},\mathbb{C}\}$ is established on the DAG-based blockchain, where each aforementioned entity is a node in this system. Depending on their abilities towards computation and storage, they can be split into two categories: light node and full node. Light nodes have limited computational power and memory space, thereby they are only responsible for generating transactions together with full nodes together. Only part of their own information is stored so that they can check it conveniently. Full nodes usually possess powerful servers with multiple functions, which can issue new transactions by finding a valid nonce and validating the previous tips. Moreover, they are responsible for storing and maintaining the entire blockchain, namely the tangle.

The architecture of our BCS system is shown in Fig. \ref{fig2}. In our BCS system $\mathbb{B}$, each EV $V_i\in \mathbb{V}$ works as a light node, and each CS $C_j\in\mathbb{C}$ works as a full node. The manager $M$ is a specific full node that manages the entire system to make sure it works instead of storing and maintaining the tangle. For the full nodes, the difficulty of hashing puzzle can be set by modifying different target values dynamically to adapt to their own computational power. For the manager, in addition to the above-mentioned function that carries out constrained multi-item double auction between EVs and CSs as an auctioneer, it has the right to add or delete nodes in real-time according to the actual situation. For example, it can permit new CSs to join this system and remove some malicious EVs from this system. Besides, it can block those invalid requests from the nodes within the system and prevent various attacks from the nodes outside the system in advance. Also, the architecture of our BCS system is based on the DAG-based blockchain, which is not only distributed and reliable to defend against attacks but also improves the throughput to be qualified for high-frequency energy trading.

In our BCS system, we use an asymmetric cryptography, such as the elliptic curve digital signature algorithm \cite{johnson2001elliptic}, for system initialization. Each EV $V_i\in\mathbb{V}$ registers on a trusted authority to be a legitimate node through obtaining a unique identification $ID_i$ that is associated with its license plate and a certificate $Cer_i$ that is signed by the private key of authority to certify the authenticity of this identity. After verifying its certificate, the EV $V_i$ can join this system and be assigned with a public/private key pair $(PK_i,SK_i)$ where its public key works as a pseudonym that is open to all nodes and its private key that is kept by itself. In asymmetric encryption, the message encrypted by the public key can be decrypted by a corresponding private key, and vice versa. After joining this system, there is a set of wallet address $\{W_{i(k)}\}_{k=1}^\theta$ owned by the EV $V_i$, where we assume there are $\theta$ wallets for each entity. Thus, the account of each EV $V_i\in\mathbb{V}$ that joins this system can be denoted by $AE_i=\{ID_i,Cer_i,(PK_i,SK_i),\{W_{i(k)}\}_{k=1}^\theta\}$. Similarly, the account of each CS $C_j\in\mathbb{C}$ can be denoted by $CE_j=\{ID_j,Cer_j,(PK_j,SK_j),\{W_{j(k)}\}_{k=1}^\theta\}$.

\begin{figure}[!t]
	\centering
	\includegraphics[width=2.5in]{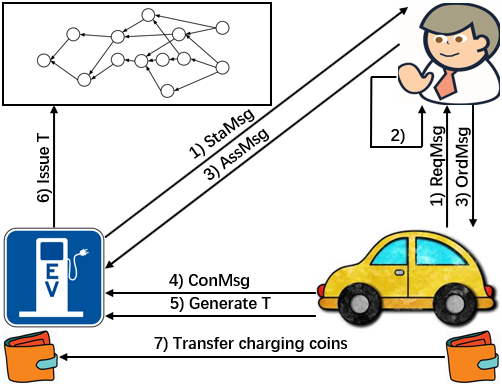}
	\centering
	\caption{The procedure and message flow of charging scheduling between CSs and EVs in a smart area.}
	\label{fig3}
\end{figure}

\subsection{Charging Scheduling Framework}
The detailed procedure and message flow of charging scheduling between EVs and CSs are shown in Fig. \ref{fig3}. We assume it performs in discrete times. Here, we only consider the situation in the time slot $t\in\{1,2,\cdots\}$, thus we ignore the time superscript $t$ in the following variables. By adding this superscript $t$, we can design online scheduling further according to the actual demand. The specific operations can be divided into 7 steps in detail:

\textbf{1) Request and Status: }Each EV $V_i\in\mathbb{V}$ sends a request message $R_i$ that includes its bid for charging at each CS $C_j\in\mathbb{C}$ to the manager. This request message is denoted by $ReqMsg=\{PK_M(SK_i(R_i)),Cer_i,STime\}$ where $PK_M$ is the public key of the manager and $STime$ is the timestamp of this message generation. At the same time, each CS $C_j\in\mathbb{C}$ sends a status message $S_j$ that includes its ask for serving a vehicle and the number of available charging piles to the manager. This status message is denoted by $StaMsg=\{PK_M(SK_j(S_j)),Cer_j,STime\}$. Here, the request and status message are encrypted by the manager's public key $PK_M$ since they are only allowed to be read by the manager for privacy protection and fair trading.

\textbf{2) Scheduling: }The manager waits to receive request messages from EVs and status messages from CSs. After collecting them, the manager confirms their legitimate identity by verifying their certificates. Then, it works as a scheduler to assign each winning EV to a CS that has unoccupied charging piles. Besides, it determines the price charged to EV and payment rewarded to CS as well, which is executed by the built-in smart contract. The smart contract is implemented by the constrained multi-item double auction model explained in the following Sec.\uppercase\expandafter{\romannumeral5} and \uppercase\expandafter{\romannumeral6}.

\textbf{3) Order and Assignment: }The manager sends an order message $O_i$ to each winning EV $V_i\in\mathbb{V}$ and an assignment message $A_j$ to each winning CS $C_j\in\mathbb{C}$. The $O_i$ includes a CS that can serve to it and the price charged to it, which is denoted by $OrdMsg=\{PK_i(SK_M(O_i)),Cer_M,STime\}$. The $A_j$ includes an assignment, the set of EVs that can be charged in $C_j$ and the payment rewarded to it, which is denoted by $AssMsg=\{PK_j(SK_M(A_j)),Cer_M,STime\}$. Here, the order and assignment message are encrypted by their public key $PK_i$ and $PK_j$ because of permitting to be read by themselves.

\textbf{4) Confirm: }If EV $V_i$ receives an order message from the manager, it implied it can be charged at the CS $C_x$ designated by the manager. Then, the EV $V_i$ sends a confirm message $F_i$ like ``I will come on time'' to the designated CS $C_x$. It is denoted by $ConMsg=\{PK_x(SK_i{F_i}),Cer_i,STime\}$ encrypted by $C_x$'s public key $PK_x$ for similar reasons.

\textbf{5) Charging: }Once CS $C_x$ receives the confirm message from the EV $V_i$, it will check whether the $V_i$ is in its assignment $A_x$. If yes, the CS $C_x$ can provide charging service to EV $V_i$ before the deadline. After charging, the EV $V_i$ generates a new transaction $T_{ix}$ according to their trading information, signs, and sends the $SK_i(T_{ix})$ back to $C_x$.

\textbf{6) Transactions: }When CS $C_x$ receives the new transaction from the EV $V_i$, it will check and sign this transaction by its private key as well. Now, the CS $C_x$ issues this new transaction $SK_x(SK_i(TX_{ix}))$ in the DAG-based blockchain. It is able to adjust the difficulty of the hashing puzzle by setting different targets dynamically according to its computational power and transaction frequency.

\textbf{7) Verification and payment: }After the transaction $T_{ix}$ is issued, it will be verified to be legal in the future when its cumulative weight is large enough. At this moment, charging coins should be transferred from the wallet of $V_i$ to $C_x$. The coins with the price charged to $V_i$ will be deducted from the wallet $W_{i(k)}$ and coins with the payment rewarded to $C_x$ will be added into the wallet $W_{x(k)}$ permanently.

\section{Problem Formulation}
In our BCS system $\mathbb{B}=\{M,\mathbb{V},\mathbb{C}\}$, CSs in $\mathbb{C}$ provide charging piles for EVs in $\mathbb{V}$ that need charging. Each CS $C_i\in\mathbb{C}$ has limited charging piles, where the number of charging piles in this charging station is $k_i\in\mathbb{Z}_+$. In general, CSs are distributed evenly across this smart area, also those located in the area center are usually more crowded. Furthermore, CSs have different charging efficiencies, where higher efficiency means shorter charging time. Thus, there are two critical attributes, location and efficiency, associated with each CS, which determine the valuation of an EV toward it. The valuation of an EV toward a CS can be decided according to its requirement. For example, when the battery of an EV is very low, it values high a CS that is nearest to it. But for an EV in a hurry, it considers both location and efficiency to minimize its charging time. In the trading between EVs and CSs, we aim to incentivize CSs to provide charging services and meet the demands of EVs. To benefit both EVs and CSs, we design a constrained multi-item double auction model that gets a truthful assignment between EVs and CSs.

\subsection{Constrained Multi-Item Double Auction Model}
Shown as Fig. \ref{fig3}, we assume that this system runs in discrete times. At each time step, EVs send request messages and CSs send status messages privately to the managers. Based on the single-round multi-item double auction model \cite{yang2011truthful}, EVs are buyers and CSs are sellers in this auction. The manager $M\in\mathbb{B}$ works as the trusted third auctioneer to assign $n$ buyers to $m$ sellers and determine the price charged to each buyer and payment rewarded to each seller.

The set of buyers is $\mathbb{V}=\{V_1,V_2,\cdots,V_n\}$ and the set of sellers is $\mathbb{C}=\{C_1,C_2,\cdots,C_m\}$. For each buyer $V_i\in\mathbb{V}$, its bid vector is denoted by $\vec{B}_i=(b_i^1,b_i^2,\cdots,b_i^m)$ where $b_i^j$ is the unit bid (maximum buying price per unit charging) of $V_i$ for charging at seller $C_j\in\mathbb{C}$. Additionally, we define a charging vector $\vec{R}=(r_1,r_2,\cdots,r_n)$ where $r_i$ is the charging amount of $V_i$. For the sellers in $\mathbb{C}$, the ask vector is denoted by $\vec{A}=(a_1,a_2,\cdots,a_m)$ where $a_j\in\vec{A}$ is the unit ask (minimum selling price per unit charging) of $C_j$. As for the number of available charging piles in each CS, we define a vector $\vec{K}=(k_1,k_2,\cdots,k_m)$ where $k_j\in\mathbb{Z}_+$ is the number of piles that can charge EVs in $C_j\in\mathbb{C}$. Here, we notice that the bids of a buyer vary with sellers since each EV has different evaluations on CSs according to its requirements regarding location and efficiency. However, the ask of a seller remains unchanged among buyers because it is only concerned about payment from charging vehicles.

By aforementioned definitions, the request message sent by buyer $V_i$ is denoted by $R_i=(\vec{B}_i,r_i)$ and the status message sent by seller $C_j$ is denoted by $S_j=(a_j,k_j)$. After it gets the collection $(\vec{B},\vec{R},\vec{A},\vec{K})$ where $\vec{B}=(\vec{B}_1;\vec{B}_2;\cdots;\vec{B}_n)$, the auctioneer determines the winning buyer set $\mathbb{V}_w\subseteq\mathbb{V}$, the winning seller set $\mathbb{C}_w\subseteq\mathbb{C}$, a mapping from $\mathbb{V}_w$ to $\mathbb{C}_w$ that is $\sigma:\{i:V_i\in\mathbb{V}_w\}\rightarrow\{j:C_j\in\mathbb{C}_w\}$, the unit price $\hat{p}_i$ charged to buyer $V_i\in\mathbb{V}_w$, and the unit payment $\bar{p}_j$ rewarded to seller $C_j$. The assignment $A_j$ for each $C_j\in\mathbb{C}_w$ is
\begin{equation}
	A_j=\left\{V_i\in\mathbb{V}_w:\sigma(i)=j\right\}\text{ where }|A_j|\leq k_j
\end{equation}
because the CS $C_j$ permits at most $k_j$ EVs to be charged at the same time, which is the reason why this model is called ``constrained'' double auction. Moreover, for each buyer $V_i\in\mathbb{V}$, its valuation vector is denoted by $\vec{V}_i=(v_i^1,v_i^2,\cdots,v_i^m)$ where $v_i^j$ is its unit valuation of $V_i$ for charging at seller $C_j\in\mathbb{C}$. For the sellers in $\mathbb{C}$, the cost vector is denoted by $\vec{C}=(c_1,c_2,\cdots,c_m)$ where $c_j\in\vec{C}$ is the unit cost of $C_j$ to provide charging service. Based on the buyer's valuation and seller's cost, the utility $\hat{u}_i$ of winning buyer $V_i\in\mathbb{V}_w$ and the utility $\bar{u}_j$ of winning seller $C_j\in\mathbb{C}_w$ can be defined as follows:
\begin{flalign}
&\hat{u}_i=(v_i^{\sigma(i)}-\hat{p}_i)\cdot r_i\\
&\bar{u}_j=(\bar{p}_j-c_j)\cdot\sum\nolimits_{V_i\in A_j}r_i
\end{flalign}
Otherwise, for losing buyer $V_i\notin\mathbb{V}_w$ and losing seller $C_j\notin\mathbb{V}_w$, their utilities are $\hat{u}_i=0$ and $\bar{u}_j=0$. Here, the utility $\hat{u}_i$ is proportional to the difference between its valuation and charged price, which implies the satisfaction level of $V_i$ on its assigned CS. The utility $\bar{u}_j$ is proportional to the difference between rewarded payment and its cost, which characterizes the profitability of $C_j$ for providing charging service.

\subsection{Design Rationales}
The constrained multi-item double auction model defined in the last subsection can be denoted by $\Psi=(\mathbb{V},\mathbb{C},\vec{B},\vec{R},\vec{A},\vec{K})$. A valid double auction mechanism has to meet the following three properties first, they are

\begin{itemize}
	\item \textbf{Individual Rationality: }The price charged to the winning buyer is not more than its bid and the payment rewarded to the winning seller is not less than its ask. Consider our model $\Psi$, we have $\hat{p}_i\leq b_i^{\sigma(i)}$ for each $V_i\in\mathbb{V}_w$ and $\bar{p}_j\geq a_j$ for each $C_j\in\mathbb{C}_w$.
	\item \textbf{Budget Balance: }The total price charged to all winning buyers is not less than the total payment rewarded to all winning sellers, which ensures the profitability of the auctioneer. Thus,  
	\begin{equation}
	\sum\nolimits_{V_i\in\mathbb{V}_w}\hat{p}_i\cdot r_i-\sum\nolimits_{C_j\in\mathbb{C}_w}\bar{p}_j\cdot\sum\nolimits_{V_i\in A_j}r_i\geq 0
	\end{equation}
	Besides, in our model $\Psi$, we need to guarantee that the assignment for each winning seller is not more than its number of charging piles, that is $|A_j|\leq k_j$ for $C_j\in\mathbb{C}_w$.
	\item \textbf{Computational Efficiency: }The auction results, including winning buyers, winning sellers, mapping from winning, price charged to buyer, and payment rewarded to seller, can be obtained in polynomial time.
\end{itemize}

In addition to the above three properties, there are two more important properties that should be satisfied strictly or approximately.
\begin{itemize}
	\item \textbf{Truthfulness: }A double auction is truthful if every buyer (resp. seller) bids (resp. asks) truthfully is one of its dominant strategies that can maximize its utility. That is to say, no buyer can increase its utility by giving a bid that is different from its true valuation and no seller can increase its utility by giving an ask that is different from its true cost. Consider our model $\Psi$, we have $\hat{u}_i$ can be maximized by bidding $\vec{B}_i=\vec{V}_i$ for each $V_i\in\mathbb{V}$ and $\bar{u}_j$ can be maximized by asking $a_j=c_j$ for each $C_j\in\mathbb{C}$ when other players do not change their strategies.
	\item \textbf{System Efficiency: }There are a number of different metrics to evaluate the system efficiency of a double auction model. The most common approach \cite{parkes2001achieving} is the number of completed trades, that is the number of buyers in the winning buyer set $\mathbb{V}_w$ in our model $\Psi$. Here, each buyer $V_i\in\mathbb{V}_w$ will be assigned to a seller $C_j\in\mathbb{C}_w$, which satisfies the requirements of both buyer and seller. To maximize it, this is in line with our original intention of designing this system to make as many EVs as possible to be charged. Other metrics, such as total price charged to winning buyers, total payment rewarded to winning seller, and profit of auctioneer, should be considered as well based on needs. We will analyze them later.
\end{itemize}

To the truthfulness, we assume the submitted vector $\vec{R}$ and $\vec{K}$ are trusted and cannot be tampered with because they can be monitored by reliable hardware. Thereby, we only consider the bids $\vec{B}$ and asks $\vec{A}$ to analyze the truthfulness of model $\Psi$. When it is truthful, the double auction model avoids being manipulated maliciously due to the fact that each player can get the best utility by telling the truth. There is no player that has the motivation to lie since they do not have to adapt to others' strategies by telling the lie for improving their utilities. Therefore, truthfulness simplifies the strategic decisions for players and makes sure a fair market environment, which plays an important role in mechanism design.

\section{Truthful Mechanism for Charging}
In this section, we design a Truthful Mechanism for Charging (TMC) based on our constrained multi-item double auction model and analyze whether it satisfies the desired properties mentioned in Sec. \uppercase\expandafter{\romannumeral5}-B.

\begin{algorithm}[!t]
	\caption{\text{TMC $(\mathbb{V},\mathbb{C},\vec{B},\vec{R},\vec{A},\vec{K})$}}\label{a1}
	\begin{algorithmic}[1]
		\renewcommand{\algorithmicrequire}{\textbf{Input:}}
		\renewcommand{\algorithmicensure}{\textbf{Output:}}
		\REQUIRE $\mathbb{V},\mathbb{C},\vec{B},\vec{R},\vec{A},\vec{K}$
		\ENSURE $\mathbb{V}_w,\mathbb{C}_w,\sigma,\hat{\mathbb{P}}_w,\bar{\mathbb{P}}_w$
		\STATE $(\mathbb{V}_c,\mathbb{C}_c,a_{j_\psi})\leftarrow$ TMC-WCD $(\mathbb{V},\mathbb{C},\vec{B},\vec{A})$
		\STATE $(\mathbb{V}_w,\mathbb{C}_w,\sigma,\hat{\mathbb{P}}_w,\bar{\mathbb{P}}_w)\leftarrow$ TMC-AP $(\mathbb{V}_c,\mathbb{C}_c,a_{j_\psi},\vec{R},\vec{K})$
		\RETURN $(\mathbb{V}_w,\mathbb{C}_w,\sigma,\hat{\mathbb{P}}_w,\bar{\mathbb{P}}_w)$ 
	\end{algorithmic}
\end{algorithm}

\begin{algorithm}[!t]
	\caption{\text{TMC-WCD $(\mathbb{V},\mathbb{C},\vec{B},\vec{A})$}}\label{a2}
	\begin{algorithmic}[1]
		\renewcommand{\algorithmicrequire}{\textbf{Input:}}
		\renewcommand{\algorithmicensure}{\textbf{Output:}}
		\REQUIRE $\mathbb{V},\mathbb{C},\vec{B},\vec{A}$
		\ENSURE $\mathbb{V}_c,\mathbb{C}_c,a_{j_\phi}$
		\STATE $\mathbb{V}_c\leftarrow\emptyset,\mathbb{C}_c\leftarrow\emptyset$
		\STATE Construct a set $\mathbb{V}'=\{V_{st}:b_s^t>0,V_s\in\mathbb{V}\}$ based on $\vec{B}$
		\STATE Sort the buyers in $\mathbb{V}'$ and transfer it to an order list $\mathbb{V}'=\langle V_{s_1t_1},V_{s_2t_2},\cdots,V_{s_xt_x}\rangle$ such that $b_{s_1}^{t_1}\geq b_{s_2}^{t_2}\geq\cdots\geq b_{s_x}^{t_x}$
		\STATE Sort the sellers, get an order list $\mathbb{C}'=\langle C_{j_1},C_{j_2},\cdots,C_{j_m}\rangle$ such that $a_{j_1}\leq a_{j_2}\leq\cdots\leq a_{j_m}$
		\STATE Find the median ask $a_{j_\phi}$ of $\mathbb{C}'$, $\phi=\left\lceil\frac{m+1}{2}\right\rceil$
		\STATE Find the minimum $\varphi$ from $\mathbb{V}'$ such that $b_{s_{\varphi+1}}^{t_{\varphi+1}}<a_{j_\phi}$
		\STATE $\mathbb{V}''\leftarrow\langle V_{s_1t_1},V_{s_2t_2},\cdots,V_{s_\varphi t_\varphi}\rangle$
		\FOR {each $V_{st}\in\mathbb{V}''$}
		\IF {$a_t<a_{j_\phi}$}
		\STATE $\mathbb{V}_c\leftarrow\mathbb{V}_c\cup\{V_{st}\}$
		\IF {$C_t\notin\mathbb{C}_c$}
		\STATE $\mathbb{C}_c\leftarrow\mathbb{C}_c\cup\{C_{t}\}$
		\ENDIF
		\ENDIF
		\ENDFOR
		\RETURN $(\mathbb{V}_c,\mathbb{C}_c,a_{j_\phi})$
	\end{algorithmic}
\end{algorithm}

\begin{algorithm}[!t]
	\caption{\text{TMC-AP $(\mathbb{V}_c,\mathbb{C}_c,a_{j_\phi},\vec{R},\vec{K})$}}\label{a3}
	\begin{algorithmic}[1]
		\renewcommand{\algorithmicrequire}{\textbf{Input:}}
		\renewcommand{\algorithmicensure}{\textbf{Output:}}
		\REQUIRE $\mathbb{V}_c,\mathbb{C}_c,a_{j_\phi},\vec{R},\vec{K}$
		\ENSURE $\mathbb{V}_w,\mathbb{C}_w,\sigma,\hat{\mathbb{P}}_w,\bar{\mathbb{P}}_w$
		\STATE $\mathbb{V}_w\leftarrow\emptyset,\mathbb{C}_w\leftarrow\emptyset,\hat{\mathbb{P}}_w\leftarrow\emptyset,\bar{\mathbb{P}}_w\leftarrow\emptyset$
		\STATE Sort the buyers in $\mathbb{V}_c$, get an ordered queue $\mathbb{Q}_c=\langle V_{s_1t_1},V_{s_2t_2},\cdots,V_{s_yt_y}\rangle$ such that $b_{s_1}^{t_1}\cdot r_{s_1}\geq b_{s_2}^{t_2}\cdot r_{s_2}\geq\cdots\geq b_{s_y}^{t_y}\cdot r_{s_y}$
		\STATE Create a tentative set $\mathbb{H}_j\leftarrow\emptyset$ from each $C_j\in\mathbb{C}_c$
		\STATE $\vec{K}'=(k'_1,k'_2,\cdots,k'_m)$ copied from vector $\vec{K}$
		\WHILE {$\mathbb{Q}_c\neq\emptyset$}
		\STATE $V_{s_lt_l}\leftarrow\mathbb{Q}_c$.pop(0) // \textit{Obtain the first element in the queue $\mathbb{Q}_c$, then remove it from $\mathbb{Q}_c$.}
		\IF {$k'_{t_l}>0$}
		\STATE $\mathbb{H}_{t_l}\leftarrow\mathbb{H}_{t_l}\cup\{V_{s_l}\},k'_{t_l}\leftarrow k'_{t_l}-1$
		\STATE $\hat{p}_{s_lt_l}\leftarrow a_{j_\phi}$
		\ELSE
		\FOR {each $V_s\in\mathbb{H}_{t_l}$}
		\STATE $\hat{p}_{st_l}\leftarrow\max\{a_{j_\phi},b_{s_l}^{t_l}\cdot(r_{s_l}/r_s)\}$
		\ENDFOR
		\FOR {each $V_{st_l}\in\mathbb{Q}_c$}
		\STATE $\mathbb{Q}_c\leftarrow\mathbb{Q}_c\backslash\{V_{st_l}\}$
		\ENDFOR
		\ENDIF
		\ENDWHILE
		\STATE $\mathbb{V}_w\leftarrow\{V_s:\exists_j(C_j\in\mathbb{C}_c\land V_s\in\mathbb{H}_j)\}$
		\FOR {each $V_s\in\mathbb{V}_w$}
		\STATE $\mathbb{I}_s\leftarrow\{C_t:V_s\in\mathbb{H}_t,C_t\in\mathbb{C}_c\}$
		\STATE Find $C_{t_m}\in\arg\max_{C_t\in\mathbb{I}_s}\{\hat{u}_{st}=(v_s^{t}-\hat{p}_{st})\cdot r_s\}$
		\STATE $\sigma(s)=t_m$
		\STATE $\hat{p}_s\leftarrow\hat{p}_{st_m},\hat{\mathbb{P}}_w\leftarrow\hat{\mathbb{P}}_w\cup\{\hat{p}_s\}$
		\IF {$C_{t_m}\notin\mathbb{C}_w$}
		\STATE $\mathbb{C}_w\leftarrow\mathbb{C}_w\cup\{C_{t_m}\}$
		\ENDIF
		\ENDFOR	
		\FOR {each $C_t\in\mathbb{C}_w$}
		\STATE $\bar{p}_{t}\leftarrow a_{j_\phi},\bar{\mathbb{P}}_w\leftarrow\bar{\mathbb{P}}_w\cup\{\bar{p}_{t}\}$ // \textit{Payments}
		\ENDFOR
		\RETURN $(\mathbb{V}_w,\mathbb{C}_w,\sigma,\hat{\mathbb{P}}_w,\bar{\mathbb{P}}_w)$
	\end{algorithmic}
\end{algorithm}

\subsection{Algorithm Design}
The process of TMC is shown in Algorithm \ref{a1}, where it is composed of two sub-processes, Winning Candidate Determination (TMC-WCD) shown in Algorithm \ref{a2} and Assignment \& Pricing (TMC-AP) shown in Algorithm \ref{a3}. In TMC, we select the winning candidates, then assign winning buyer candidates to winning seller candidates truthfully. At the same time, the price charged to winning buyer and payment rewarded to winning seller can be determined.

Shown as Algorithm \ref{a2}, in the process of winning candidate determination, we sort the buyers in descending order based on their bids for different sellers, where each buyer $V_s\in\mathbb{V}$ is replaced with a buyer set $\{V_{st}:b_s^t>0,C_t\in\mathbb{C}\}$ in which buyer $V_s$ gives a positive bid to seller $C_t$. The sellers are sorted in ascending order based on their ask. The median of asks from sellers $a_{j_\phi}$ is selected as a threshold \cite{jin2015auction1} to control the number of buyer and seller candidates. Let $\varphi$ satisfy $b_{s_{\varphi}}^{t_{\varphi}}\geq a_{j_\phi}$ and $b_{s_{\varphi+1}}^{t_{\varphi+1}}<a_{j_\phi}$, buyer $V_{st}$ will be a winning buyer candidate if its bid $b_s^t$ is not less than $b_{s_{\varphi}}^{t_{\varphi}}$ and the ask of requested seller $a_t$ is less than the threshold $a_{j_\phi}$. Seller $C_t$ will be a winning seller candidate if its ask $a_t$ is less than $a_{j_\phi}$ and there exists at least one winning buyer candidate bidding for it. 

Shown as Algorithm \ref{a3}, in the process of assignment \& pricing, we sort the $\mathbb{V}_c$ in descending order based on their total bids, which is the unit bids multiplied by their charging amounts. The total bid of buyer $V_{st}$ is $b_s^t\cdot r_s$ definitely. For each winning buyer candidate $V_{st}\in\mathbb{V}_c$, it has met the basic conditions for closing a deal because there is a winning seller candidate $C_t\in\mathbb{C}_c$ with $b_s^t>a_t$ and $a_t<a_{j_\phi}$. For each seller $C_t\in\mathbb{C}_c$, we assign the buyers with larger total bids to it in priority. Thereby, there is a ``tentative set'' $\mathbb{H}_t$ associated with each $C_t\in\mathbb{C}_c$, which contains at most $k_t$ buyers with maximum total bids to $C_t$ in $\mathbb{Q}_c$. It can be implemented in line 5-18. We denoted by $\hat{p}_{st}$ the unit price charged to buyer $V_s$ that gets service from $C_t$. Similar, the utility $\hat{u}_{st}$ for each buyer $V_s\in\mathbb{H}_t$ can be defined as $\hat{u}_{st}=(v_s^t-\hat{p}_{st})\cdot r_s$; Otherwise, the utility is $\hat{u}_{st}=0$ for the buyer $C_s\notin\mathbb{H}_t$. For example, let $\langle V_{o_1t},\cdots,V_{o_zt}\rangle\subseteq\mathbb{Q}_c$ be all buyers in $\mathbb{Q}_c$ that bid for $C_t$ with $b_{o_1}^{t}\cdot r_{o_1}\geq\cdots\geq b_{o_z}^{t}\cdot r_{o_z}$. If $k_t\geq z$, we have $\mathbb{H}_t=\langle V_{o_1t},\cdots,V_{o_zt}\rangle$ and $\hat{p}_{o_it}=a_{j_\phi}$ for each $V_{o_it}\in\mathbb{H}_t$; Else, we have $\mathbb{H}_t=\langle V_{o_1t},\cdots,V_{o_{k_t}t}\rangle$ and $\hat{p}_{o_it}=\max\{a_{j_\phi},b_{o_{k_t+1}}^{t}\cdot(r_{o_{k_t+1}}/r_{o_i})\}$ for each $V_{o_it}\in\mathbb{H}_t$ to guarantee truthfulness. Then, for each winning buyer $V_s\in\mathbb{V}_w$, it can be assigned to one of the seller in $\mathbb{I}_s$. The $V_s$ selects the optimal seller $C_{t_m}\in\mathbb{I}_s$ such that $\hat{u}_{st_m}\geq\hat{u}_{st}$ for each $C_t\in\mathbb{I}_s$. Now, the mapping is $\sigma(s)=t_m$ and the charged price is $\hat{p}_s=\hat{p}_{st_m}$. Finally, the payment rewarded to winning seller in $\mathbb{C}_w$ is given by $a_{j_\phi}$ unanimously.

\subsection{Properties of TMC}
Next, we argue that our proposed TMC mechanism satisfies individual rationality, budget balance, computational efficiency, and truthfulness.
\begin{lem}
	The TMC is individually rational.
\end{lem}
\begin{proof}
	For each winning buyer $V_s\in\mathbb{V}_w$ and winning seller $C_t\in\mathbb{C}_w$, we need to show that the charged price $\hat{p}_s\leq b_s^{\sigma(s)}$ and the rewarded payment $\bar{p}_t\geq a_t$. According to Algorithm \ref{a2}, it must be $a_t< a_{j_\phi}$ if $C_t\in\mathbb{C}_w\subseteq\mathbb{C}_c$. Thereby the payment rewarded to winning seller $\bar{p}_t=a_{j_\phi}$ is larger than its ask $a_t$ for each $C_t\in\mathbb{C}_w$. Consider a winning buyer $V_s\in\mathbb{V}_w$, the charged price is either $\hat{p}_s=a_{j_\phi}$ or $\hat{p}_s=b_{s_l}^{\sigma(s)}\cdot(r_{s_l}/r_s)$.
	\begin{itemize}
		\item In the first case, it must be $b_s^{\sigma(s)}\geq b_{p_{\varphi}}^{q_{\varphi}}$ if $V_s\in\mathbb{V}_w$. The price charged to winning buyer $\hat{p}_s=a_{j_\phi}$ is not more than its bid $b_s^{\sigma(s)}$ definitely.
		\item In the second case, we have $b_s^{\sigma(s)}\cdot r_s\geq b_{s_l}^{t_l}\cdot r_{s_l}$ according to Algorithm \ref{a3}. The price charged to winning buyer $\hat{p}_s=b_{s_l}^{\sigma(s)}\cdot(r_{s_l}/r_s)$ is not more than its bid $b_s^{\sigma(s)}$.
	\end{itemize}
	Thus, the TMC is individually rational because all winning buyers and sellers are individually rational.
\end{proof}

\begin{lem}
	The TMC is budget balanced.
\end{lem}
\begin{proof}
	Shown as Algorithm \ref{a3}, each winninng buyer $V_s\in\mathbb{V}_w$ is assigned to exact one winning seller $C_t\in\mathbb{C}_w$, hence this is a many-to-one mapping from $\mathbb{V}_w$ to $\mathbb{C}_w$ and $\mathbb{V}_w=\cup_{C_j\in\mathbb{C}_w}A_j$. For each mapping $\sigma(s)=t$ from winning buyer $V_s$ to winning seller $C_t$, we have $\hat{p}_s\geq a_{j_\phi}=\bar{p}_t$. Based on (5), it can be shown that
	\begin{equation}
		\sum\nolimits_{V_i\in\mathbb{V}_w}(\hat{p}_i-\bar{p}_{\sigma(i)})\cdot r_i\geq 0
	\end{equation}
	Besides, it is easy to see that the assignment $A_j$ for each winning seller $C_j\in\mathbb{C}_w$ satisfies $|A_j|\leq k_j$.
\end{proof}

\begin{lem}
	The TMC is truthful.
\end{lem}
\begin{proof}
	Here, we need to show the truthfulness to sellers and buyers one by one as follows:
	
	For each buyer $V_s\in\mathbb{V}$, we need to show its utility $\hat{u}_s$ when giving a truthful bid $\vec{B}_s=\vec{V}_s$ is not less than its corresponding utility $\hat{u}'_s$ when given an untruthful bid $\vec{B}'_s\neq\vec{V}_s$. Afterward, any notation $x_s$ and $x'_s$ refer to the concepts given by bid $\vec{B}_s$ and $\vec{B}'_s$ respectively.
	
	\noindent
	\textbf{(a)} The $V_s$ is a winning buyer when bidding truthfully: According to Algorithm \ref{a3}, it wins the seller $C_{t_m}$ such that maximizes its utility. Thereby, the utility $\hat{u}_{st_m}\geq\hat{u}_{st}$ for each $C_t\in\mathbb{I}_s$. First, for each seller $C_t\in\mathbb{I}_s$, it implies $V_s\in\mathbb{H}_t$ and giving an untruthful bid $(b_s^t)'$ to seller $C_t$ cannot increase the utility such that $\hat{u}'_{st}>\hat{u}_{st}$ since
	\begin{itemize}
		\item $(b_s^t)'>v_s^t$: The charged price $\hat{p}'_{st}(=\hat{p}_{st})$ will not be changed because the $V_s$ has been in $\mathbb{H}_t$ when bidding $b_s^t(=v_s^t)$. Thus, we have $\hat{u}'_{st}=\hat{u}_{st}$.
		\item $(b_s^t)'<v_s^t$: The charged price $\hat{p}'_{st}(=\hat{p}_{st})$ will not be changed if the $V_s$ is still in $\mathbb{H}_t$ when bidding $(b_s^t)'(<v_s^t)$. Thus, we have $\hat{u}'_{st}=\hat{u}_{st}$. However, when this untruthful bid $(b_s^t)'$ decreases to be lower than  $\hat{p}_{st}$, the $V_s$ cannot be in $\mathbb{H}_t$ and $\hat{u}'_{st}=0$. Thus, we have $\hat{u}'_{st}<\hat{u}_{st}$.
	\end{itemize}
	Then, for each seller $C_t\in\mathbb{V}\backslash\mathbb{I}_s$, it can be divided into two sub-cases where giving an untruthful bid $(b_s^t)'$ to seller $C_t$ cannot increase the utility such that $\hat{u}'_{st}>\hat{u}_{st}$ as well. They are analyzed as follows:
	\begin{itemize}
		\item $V_{st}\notin\mathbb{V}_c$: If $a_t\geq a_{j_\phi}$, it is impossible to make $V_{st}$ be in $\mathbb{V}_c$ according to Algorithm \ref{a2} regardless of what the $(b_s^t)'$ is. Thus, we have $\hat{u}'_{st}=\hat{u}_{st}=0$. If $a_t<a_{j_\phi}$ but truthful bid $b_s^t(=v_s^t)<a_{j_\phi}$, the $V_s$ has to increase its bid such that $(b_s^t)'\geq a_{j_\phi}$ in order to make $V_{st}$ be in $\mathbb{V}_c$. At this time, we have
		\begin{equation}
			\hat{u}'_{st}=(v_s^t-\hat{p}'_{st})\cdot r_s\leq (v_s^t-a_{j_\phi})\cdot r_s<0
		\end{equation}
		if $C_t\in\mathbb{I}_s$ when bidding $(b_s^t)'$ untruthfully, otherwise $\hat{u}'_{st}=0$. Thus, we have $\hat{u}'_{st}<\hat{u}_{st}=0$.
		\item $V_{st}\in\mathbb{V}_c$ but $V_{s}\notin\mathbb{H}_t$: In this case, we can know that the $\mathbb{H}_t$ has been full where $|\mathbb{H}_t|=k_t$. From this, we have $b_o^t\cdot r_o\geq b_s^t(=v_s^t)\cdot r_s$ where $V_o$ has minimum value of $b_o^t\cdot r_o$ among all buyers in $\mathbb{H}_t$. The $V_s$ has to increase its bid such that $(b_s^t)'\geq b_o^t\cdot(r_o/r_s)$ in order to replace $V_o$ in $\mathbb{H}_t$. Then, the charged price will be changed to $\hat{p}'_{st}=b_o^t\cdot(r_o/r_s)$. The utility is 
		\begin{equation}
			\hat{u}'_{st}=(v_s^t-\hat{p}'_{st})\cdot r_s=\left(v_s^t-\frac{b_o^tr_o}{r_s}\right)\cdot r_s<0
		\end{equation}
		since $b_o^t\cdot r_o\geq v_s^t\cdot r_s$. Thus we have $\hat{u}'_{st}<\hat{u}_{st}=0$.
	\end{itemize}
	\textbf{(b)} The $V_s$ is not a winning buyer when bidding truthfully: According to Algorithm \ref{a3}, there is no such a $\mathbb{H}_t$ for $C_t\in\mathbb{C}_c$ that has $V_s\in\mathbb{H}_t$. For each seller $C_t\in\mathbb{C}$, we have utility $\hat{u}_{st}=0$. Similarly, we need to show giving an untruthful bid $(b_s^t)'$ to seller $C_t$ cannot increase the utility such that $\hat{u}'_{st}>\hat{u}_{st}$. The analysis about it can be divided into $V_{st}\notin\mathbb{V}_c$ and $V_{st}\in\mathbb{V}_c$ but $V_{s}\notin\mathbb{H}_t$, which are the same as the analysis for $C_t\in\mathbb{V}\backslash\mathbb{I}_s$ in part (a). Thus, we have $\hat{u}'_{st}\leq\hat{u}_{st}=0$.
	
	\noindent
	From above, the utility of $V_s$ to each seller $C_t\in\mathbb{C}$ satisfies $\hat{u}_{st}\geq\hat{u}'_{st}$ definitely. By selecting the the seller such that maximizes its utility, we have $\hat{u}_s\geq\hat{u}'_s$. Therefore, the buyers are truthful.
	
	For each seller $C_t\in\mathbb{C}$, we need to show its utility $\bar{u}_t$ when giving a truthful ask $a_t=c_t$ is not less than its corresponding utility $\bar{u}'_t$ when given an untruthful ask $a'_t\neq c_t$. Afterward, any notation $x_t$ and $x'_t$ refer to the concepts given by bid $a_t$ and $a'_t$ respectively.
	
	\noindent
	\textbf{(c)} The $C_t$ is a winning seller when asking truthfully: According to Algorithm \ref{a3}, its ask $a_t(=c_t)<a_{j_\phi}$ and at least one buyer are assigned to it. Hence, we have $|A_t|>0$. Its utility can be denoted by
	\begin{equation}
		\bar{u}_t=(a_{j_\phi}-c_t)\cdot\sum\nolimits_{V_i\in A_t}r_i>0
	\end{equation}
	since $\bar{p}_t=a_{j_\phi}$. Consider giving an untruthful ask $a'_t$, we can discuss as follows:
	\begin{itemize}
		\item $a'_t\geq a_{j_\phi}$: It loses this auction because the new median ask $a'_{j_\phi}$ is not less than $a_{j_\phi}$ and $a'_t\geq a'_{j_\phi}\geq a_{j_\phi}$. Thus, the utility is $\bar{u}'_t=0<\bar{u}_t$.
		\item $a'_t<a_{j_\phi}$: At the time, the new median ask $a'_{j_\phi}$ is equal to $a_{j_\phi}$ and $a'_t<a'_{j_\phi}=a_{j_\phi}$. Moveover, the assignment of $C_t$ remains unchanged, $A'_t=A_t$, and the rewarded payment $\bar{p}'_t=\bar{p}_t$. According to (9), the utility when asking untruthfully is the same as that when asking truthfully. Thus, we have $\bar{u}'_t=\bar{u}_t$.
	\end{itemize}
	\textbf{(d)} The $C_t$ is not a winning seller when asking truthfully: At this time, its utility when asking truthfully is $\bar{u}_t=0$. Here, we need to analyze how the $C_t$ loses this auction. If losing since $c_t\geq a_{j_\phi}$, we have
	\begin{itemize}
		\item $a'_t<a_{j_\phi}$: The new median ask $a'_{j_\phi}$ is not more than $a_{j_\phi}$ and $a'_t\leq a'_{j_\phi}\leq a_{j_\phi}$. If the $C_t$ still loses the auction, its utility $\bar{u}'_t=0$. If the $C_t$ wins now, its utility $\bar{u}'_t\leq0$ because of the rewarded payment $\bar{p}'_t= a'_{j_\phi}\leq a_{j_\phi}\leq c_t$. Thus, we have $\bar{u}'_t\leq\bar{u}_t=0$.
		\item $a'_t\geq a_{j_\phi}$: It loses this auction because the new median ask $a'_{j_\phi}$ is equal to $a_{j_\phi}$ and $a'_t\geq a_{j_\phi}=a'_{j_\phi}$. Thus, the utility is $\bar{u}'_t=\bar{u}_t=0$.
	\end{itemize}
	If $c_t< a_{j_\phi}$ but still loses, there are two situations that no buyer $V_{st}$ gives a bid $b_s^t$ such that $b_s^t\geq a_{j_\phi}$ or utility $\hat{u}_{st}$ is not the maximum one for each $V_s\in\mathbb{V}_w$. Now,
	\begin{itemize}
		\item $a'_t<a_{j_\phi}$: The above two situations still happen because the new median ask $a'_{j_\phi}=a_{j_\phi}$ and $a'_t<a'_{j_\phi}$. Thus, we have $\bar{u}'_t=\bar{u}_t=0$.
		\item $a'_t\geq a_{j_\phi}$: It loses this auction because the new median ask $a'_{j_\phi}$ is not less than $a_{j_\phi}$ and $a'_t\geq a'_{j_\phi}\geq a_{j_\phi}$. Thus, the utility is $\bar{u}'_t=\bar{u}_t=0$.
	\end{itemize}
	 Therefore, the sellers are truthful.
	 
	 In summary, both buyers and sellers cannot improve the utility by deviating from their valuations and costs.
\end{proof}

\begin{algorithm}[!t]
	\caption{\text{EMC $(\mathbb{V},\mathbb{C},\vec{B},\vec{R},\vec{A},\vec{K})$}}\label{a4}
	\begin{algorithmic}[1]
		\renewcommand{\algorithmicrequire}{\textbf{Input:}}
		\renewcommand{\algorithmicensure}{\textbf{Output:}}
		\REQUIRE $\mathbb{V},\mathbb{C},\vec{B},\vec{R},\vec{A},\vec{K}$
		\ENSURE $\mathbb{V}_w,\mathbb{C}_w,\sigma,\hat{\mathbb{P}}_w,\bar{\mathbb{P}}_w$
		\STATE $(\mathbb{V}_c,\mathbb{C}_c,a_{j_\psi})\leftarrow$ EMC-WCD $(\mathbb{V},\mathbb{C},\vec{B},\vec{A})$
		\STATE $(\mathbb{V}_w,\mathbb{C}_w,\sigma,\hat{\mathbb{P}}_w,\bar{\mathbb{P}}_w)\leftarrow$ EMC-AP $(\mathbb{V}_c,\mathbb{C}_c,a_{j_\psi},\vec{R},\vec{K})$
		\RETURN $(\mathbb{V}_w,\mathbb{C}_w,\sigma,\hat{\mathbb{P}}_w,\bar{\mathbb{P}}_w)$ 
	\end{algorithmic}
\end{algorithm}

\begin{lem}
	The TMC is computationally efficient.
\end{lem}
\begin{proof}
	In Algorithm \ref{a2}, there are $nm$ buyers in $\mathbb{V}'$, hence it takes $O(nm\log(nm))$ and $O(m\log(m))$ times to sort the $\mathbb{V}'$ and $\mathbb{C}'$ respectively. The size of $\mathbb{V}''$ in line 7 is at most $n\phi$, and the number of iterations in the for-loop (line 8-15) is at most $n\phi$. Consequently, the time complexity of Algorithm \ref{a2} is $O(nm\log(nm))$. In Algorithm \ref{a3}, it takes $O(n\phi\log(n\phi))$ times to sort the $V_c$. The number of iterations in the while-loop (line 5-18) is at most $n$. The line 12 can be executed at most $\sum_{i=1}^{m}k_i$ times, thus it takes $O(n+\sum_{i=1}^{m}k_i)$ time to execute this while-loop. Besides, there are at most $n$ buyers in $\mathbb{V}_w$ and find the best (line 22) from at most $\phi$ sellers, thus it takes $O(n\phi)$ time to execute the for-loop (line 20-28). Consequently, the time complexity of Algorithm \ref{a2} is $O(n\phi\log(n\phi))$ and overall time complexity of TMC is $O(nm\log(nm))$.
\end{proof}

\begin{thm}
	The TMC is individually rational, budget balanced, truthful, and computationally efficient.
\end{thm}
\begin{proof}
	It can be derived from Lemma 1 to Lemma 4.
\end{proof}

\begin{algorithm}[!t]
	\caption{\text{EMC-AP $(\mathbb{V}_c,\mathbb{C}_c,a_{j_\phi},\vec{R},\vec{K})$}}\label{a5}
	\begin{algorithmic}[1]
		\renewcommand{\algorithmicrequire}{\textbf{Input:}}
		\renewcommand{\algorithmicensure}{\textbf{Output:}}
		\REQUIRE $\mathbb{V}_c,\mathbb{C}_c,a_{j_\phi},\vec{R},\vec{K}$
		\ENSURE $\mathbb{V}_w,\mathbb{C}_w,\sigma,\hat{\mathbb{P}}_w,\bar{\mathbb{P}}_w$
		\STATE $\mathbb{V}_w\leftarrow\emptyset,\mathbb{C}_w\leftarrow\emptyset,\hat{\mathbb{P}}_w\leftarrow\emptyset,\bar{\mathbb{P}}_w\leftarrow\emptyset$
		\STATE Sort the buyers in $\mathbb{V}_c$, get an ordered queue $\mathbb{Q}_c=\langle V_{s_1t_1},V_{s_2t_2},\cdots,V_{s_yt_y}\rangle$ such that $b_{s_1}^{t_1}\cdot r_{s_1}\geq b_{s_2}^{t_2}\cdot r_{s_2}\geq\cdots\geq b_{s_y}^{t_y}\cdot r_{s_y}$
		\STATE $\vec{K}'=(k'_1,k'_2,\cdots,k'_m)$ copied from vector $\vec{K}$
		\WHILE {$\mathbb{Q}_c\neq\emptyset$}
		\STATE $V_{s_lt_l}\leftarrow\mathbb{Q}_c$.pop(0) 
		\IF {$k_{t_l}>0$}
		\STATE $\sigma(s_l)=t_l,k'_{t_l}\leftarrow k'_{t_l}-1$
		\STATE $\mathbb{V}_w\leftarrow\mathbb{V}_w\cup\{V_{s_l}\},\hat{p}_{s_l}\leftarrow a_{j_\phi},\hat{\mathbb{P}}_w\leftarrow\hat{\mathbb{P}}_w\cup\{\hat{p}_{s_l}\}$
		\IF {$C_{t_l}\notin\mathbb{C}_w$}
		\STATE $\mathbb{C}_{w}\leftarrow\mathbb{C}_w\cup\{C_{t_l}\}$
		\ENDIF
		\FOR {each $V_{s_lt}\in\mathbb{Q}_c$}
		\STATE $\mathbb{Q}_c\leftarrow\mathbb{Q}_c\backslash\{V_{s_lt}\}$
		\ENDFOR
		\ELSE
		\STATE $A_{t_l}\leftarrow\{V_s\in\mathbb{V}_w:\sigma(s)=t_l\}$
		\FOR {each $V_s\in A_{t_l}$}
		\STATE $\hat{p}_s\leftarrow\max\{a_{j_\phi},b_{s_l}^{t_l}\cdot(r_{s_l}/r_s)\}\in\hat{\mathbb{P}}_w$
		\ENDFOR
		\FOR {each $V_{st_l}\in\mathbb{Q}_c$}
		\STATE $\mathbb{Q}_c\leftarrow\mathbb{Q}_c\backslash\{V_{st_l}\}$
		\ENDFOR
		\ENDIF
		\ENDWHILE
		\FOR {each $C_t\in\mathbb{C}_w$}
		\STATE $\bar{p}_{t}\leftarrow a_{j_\phi},\bar{\mathbb{P}}_w\leftarrow\bar{\mathbb{P}}_w\cup\{\bar{p}_{t}\}$ // \textit{Payments}
		\ENDFOR
		\RETURN $(\mathbb{V}_w,\mathbb{C}_w,\sigma,\hat{\mathbb{P}}_w,\bar{\mathbb{P}}_w)$ 
	\end{algorithmic}
\end{algorithm}

\section{Efficient Mechanism For Charging}
Even though TMC is able to ensure truthfulness, it sacrifices the system efficiency. Shown as Algorithm \ref{a3}, suppose the winning buyer $V_s$ satisfies $V_s\in\mathbb{H}_{t_1}$ and $V_s\in\mathbb{H}_{t_2}$, it can be assigned to only one seller $t\in\{t_1,t_2\}$. Thus, another seller will have a charging pile empty, which could be used to charge other buyers. Thus, for each winning buyer $V_s\in\mathbb{V}_w$ with $|\mathbb{I}_s|>1$, there are $|\mathbb{I}_s|-1$ charging piles being wasted. To address this drawback, we propose an Efficient Mechanism for Charging (EMC) to improve system efficiency and ensure its truthfulness to some extent.

\subsection{Algorithm Design}
The process of EMC is shown in Algorithm \ref{a4}. Similar to TMC, it is composed of Winning Candidate Determination (EMC-WCD) and Assignment \& Pricing (EMC-AP). Here, the EMC-WCD is the same as TMC-WCD shown in Algorithm \ref{a2}, which can be used to generate a winning buyer candidate set $\mathbb{V}_c$ and winning seller candidate set $\mathbb{C}_c$. The EMC-AP is shown in Algorithm \ref{a5}.

Shown as Algorithm \ref{a5}, in the process of assignment \& pricing, we sort the winning buyer candidates in $\mathbb{V}_c$ in descending order based on their total bids. Then, we give priority to assigning the buyer that can give the maximum total bid, which is the critical step to improve the system efficiency. At each iteration, we pop the buyer that has the maximum total bid from $\mathbb{Q}_c$, denoted by $V_{s_lt_l}$ and check whether the seller $C_{t_l}$ requested by $V_{s_l}$ still have available charging piles. If yes, $k'_{t_l}>0$, the buyer $V_{s_l}$ will be assigned to seller $C_{t_l}$. Also, $V_{s_l}$ is a winning buyer,  $C_{t_l}$ is a winning seller, and the price charged to buyer $V_{s_l}$ is given by $\hat{p}_{s_l}=a_{j_\phi}$ tentatively. Furthermore, the requests from buyer $V_{s_l}$ to other sellers should be deleted from $\mathbb{Q}_c$ since the buyer $V_{s_l}$ has been assigned. If no, $k'_{t_l}=0$, the buyer $V_{s_l}$ will not be assigned to seller $C_{t_l}$ because there is no unoccupied charging piles in $C_{t_l}$. It implies that in previous iterations, the buyers in $A_{t_l}$ have been assigned to seller $C_{t_l}$, then the price charged to each winning buyer $V_s\in A_{t_l}$ has to be changed to its critical price $\hat{p}_s=b_{s_l}^{t_l}\cdot(r_{s_l}/r_s)$. Finally, the payment rewarded to winning seller is given by $a_{j_\phi}$ unanimously.

\subsection{Properties of EMC}
Similar to Sec. \uppercase\expandafter{\romannumeral6}-B, we analyze whether our EMC mechanism satisfies the aforementioned four properties.

\begin{lem}
	The EMC is individually rational. 
\end{lem}
\begin{proof}
	It can be discussed similar to proof of Lemma 1.
\end{proof}

\begin{lem}
	The EMC is budget balanced. 
\end{lem}
\begin{proof}
	It can be discussed similar to proof of Lemma 2.
\end{proof}

\begin{lem}
	The EMC is not truthful, but the truthfulness is held for sellers.
\end{lem}
\begin{proof}
	For a buyer $V_s\notin\mathbb{V}_w$ when bidding truthfully, giving an untruthful bid $(b_s^t)'$ to seller $C_t$ cannot increase the utility such that $\hat{u}'_{st}<\hat{u}_{st}=0$, which can be divided into $V_{st}\notin\mathbb{V}_c$ (similar to the analysis for $V_{st}\notin\mathbb{V}_c$ in part (a) of Lemma 3) and $V_{st}\in\mathbb{V}_c$. Consider the case $V_{st}\in\mathbb{V}_c$, it exists a $V_{ot}\in\mathbb{Q}_c$ with $b_o^t\cdot r_o\geq b_s^t(=v_s^t)\cdot r_s$ which is the last one can be assigned to seller $C_t$ in Algorithm \ref{a5}. To replace $V_{ot}$, the $V_s$ has to bid a $(b_s^t)'$ such that $(b_s^t)'\geq b_o^t\cdot(r_o/r_s)$ and the charged price will be $\hat{p}'_{st}=b_o^t\cdot(r_o/r_s)$. From here, we have $\hat{u}'_{st}\leq 0$ since $\hat{p}'_{st}\geq v_s^t$. For a buyer $V_s\in\mathbb{V}_w$, we give two examples where giving an untruthful bid is possible to improve its utility. There are two sellers $o_1$ and $o_2$ requested by the $V_s$ lie in $\mathbb{Q}_c$ when it bids truthfully, thus $\mathbb{Q}_c=\langle\cdots,V_{so_1},\cdots,V_{s_lo_1},\cdots,V_{so_2},\cdots\rangle$ with $v_s^{o_1}\cdot r_s\geq b_{s_l}^{o_1}\cdot r_{s_l}\geq v_s^{o_2}\cdot r_s$. The results returned by Algorithm \ref{a5} are $\sigma(s)=o_1$, $\hat{p}_s=b_{s_l}^{o_1}\cdot(r_{s_l}/r_s)$, and $k'_{o_2}>0$. At this time, its utility is $\hat{u}_s=(v_s^{o_1}-b_{s_l}^{o_1}\cdot(r_{s_l}/r_s))\cdot r_s$. When giving an untruthful bid $(b_s^{o_1})'$ such that $(b_s^{o_1})'<v_s^{o_2}$, the results are changed to $\sigma'(s)=o_2$, $\hat{p}'_s=a_{j_\phi}$. At this time, its utility is $\hat{u}'_s=(v_s^{o_2}-a_{j_\phi})\cdot r_s$. We cannot judge which one is larger since $v_s^{o_1}\geq v_s^{o_2}$ and $b_{s_l}^{o_1}\cdot(r_{s_l}/r_s)\geq a_{j_\phi}$. If $\hat{u}'_s>\hat{u}_s$, its utility can be improved by bidding untruthfully. Similarly, it can give an untruthful bid $(b_s^{o_2})'$ such that $(b_s^{o_2})'>v_s^{o_1}$, the results are changed to $\sigma'(s)=o_2$, $\hat{p}'_s=a_{j_\phi}$ as well. Thus, truthfulness is not held for winning buyers.
	
	The analysis for sellers are similar to the case (c) and (d) in proof of Lemma 3, thus truthfulness is held for sellers.
\end{proof}

\begin{lem}
	The EMC is computationally efficient. 
\end{lem}
\begin{proof}
	It can be discussed similar to proof of Lemma 4.
\end{proof}

\begin{thm}
	The EMC is individually rational, budget balanced, and computationally efficient, but not truthful.
\end{thm}
\begin{proof}
	It can be derived from Lemma 5 to Lemma 8.
\end{proof}

In EMC, we attempt to assign each buyer in $\mathbb{V}_c$ to a seller in a greedy manner, thus avoiding the waste of charging piles. Therefore, it increases the number of winning buyers (successful trades) and improves system efficiency. Even though the winning buyers are able to improve their utilities by bidding untruthfully, this is difficult to achieve. In our BCS system, each player bids or asks privately. The buyer cannot have knowledge of other players' strategies such as other buyers' bids and sellers' asks. Thus, it is not able to predict whether it will win, which seller it will be assigned to, and its charged price. For the buyer losing the auction, it is possible to get a negative utility when bidding untruthfully. For the buyer winning the auction, despite the potential of improving its utility, there is also the possibility of losing the auction when bidding untruthfully. Obviously, they have no motivation to lie because the risks are great. Therefore, we can say the EMC is truthful to some extent.

\subsection{A Walk-through Example}
To understand our TMC and EMC algorithms clearly and compare their difference, we give a walk-through example with 5 buyers and 5 sellers. The bids and charging amounts of buyers, the asks and number of charging piles of sellers are shown in Table \ref{table1}.

In TMC-WCD (EMC-WCD) according to Algorithm \ref{a2}, the median of asks is denoted by $a_{j_\phi}=a_3=3$. We can get $\mathbb{V}'=\langle V_{31},V_{41},V_{14},V_{23},V_{32},V_{55},V_{12},V_{34},V_{42},V_{44},V_{54}\rangle$. By removing those $V_{st}\in\mathbb{V}'$ with $a_t\geq a_{j_\phi}$, we have $\mathbb{V}_c=\{V_{14},V_{32},V_{12},V_{34},V_{42},V_{44},V_{54}\}$ and $\mathbb{C}_c=\{C_2,C_4\}$. Then, sort the $\mathbb{V}_c$ according to their total bids, we have $\mathbb{Q}_c=\langle V_{32}=30,V_{14}=25,V_{34}=24,V_{12}=20,V_{42}=16,V_{44}=12,V_{54}=9\rangle$.

For TMC, in TMC-AP according to Algorithm \ref{a3}, we have tentative sets $\mathbb{H}_2=\{V_1,V_3\}$ with $\hat{p}_{12}=\max\{a_{j_\phi},b_4^2\cdot(r_4/r_1)\}=3.2,\hat{p}_{32}=\max\{a_{j_\phi},b_4^2\cdot(r_4/r_3)\}=3$ and $\mathbb{H}_4=\{V_1,V_3\}$ with $\hat{p}_{14}=\max\{a_{j_\phi},b_4^4\cdot(r_4/r_1)\}=3,\hat{p}_{34}=\max\{a_{j_\phi},b_4^4\cdot(r_4/r_3)\}=3$. For the buyer $V_1$, its utility satisfies $\hat{u}_{12}=(b_1^2-\hat{p}_{12})\cdot r_1=4<10=\hat{u}_{14}$. For the buyer $V_3$, its utility satisfies $\hat{u}_{32}>\hat{u}_{34}$. Thus, we have $\mathbb{V}_w=\{V_1,V_3\}$, $\mathbb{C}_w=\{C_2,C_4\}$, $\{\sigma(1)=4,\sigma(3)=2\}$, $\hat{P}_w=\{\hat{p}_1=3,\hat{p}_3=3\}$, and  $\bar{P}_w=\{\bar{p}_2=3,\bar{p}_4=3\}$.

\begin{table}[!t]
	\renewcommand{\arraystretch}{1.2}
	\caption{An example with 5 buyers and 5 sellers.}
	\label{table1}
	\centering
	\begin{tabular}{c|ccccc|c}
		\hline
		$\vec{B}$ & $C_1$ & $C_2$ & $C_3$ & $C_4$ & $C_5$ & $\vec{R}$\\
		\hline
		$V_1$ & 0 & 4 & 0 & 5 & 2 & 5\\
		$V_2$ & 2 & 0 & 5 & 1 & 0 & 2\\
		$V_3$ & 7 & 5 & 0 & 4 & 0 & 6\\
		$V_4$ & 6 & 4 & 0 & 3 & 0 & 4\\
		$V_5$ & 0 & 0 & 2 & 3 & 5 & 3\\
		\hline
		$\vec{A}$ & 4 & 1 & 3 & 2 & 5 & -\\
		$\vec{K}$ & 1 & 2 & 4 & 2 & 3 & -\\
		\hline
	\end{tabular}
\end{table}

For EMC, in EMC-AP according to Algorithm \ref{a5}, we assign buyer $V_3$ to seller $C_2$ with $\hat{p}_3=3$ in the first iteration, then $\mathbb{Q}_c$ is revised to $\mathbb{Q}_c=\langle V_{14}=25,V_{12}=20,V_{42}=16,V_{44}=12,V_{54}=9\rangle$. Repeat it until $\mathbb{Q}_c=\emptyset$, we have $\mathbb{V}_w=\{V_1,V_3,V_4,V_5\}$, $\mathbb{C}_w=\{C_2,C_4\}$, $\{\sigma(1)=4,\sigma(3)=2,\sigma(4)=2,\sigma(5)=4\}$, $\hat{P}_w=\{\hat{p}_1=3,\hat{p}_3=3,\hat{p}_4=3,\hat{p}_5=3\}$, and  $\bar{P}_w=\{\bar{p}_2=3,\bar{p}_4=3\}$. From this example, we can see that the winning sellers in TMC are not full where there are idle charging piles not used to charge vehicles. Therefore, the number of successful trades $|\mathbb{V}_w|=2$ in TMC is less than $|\mathbb{V}_w|=4$ in EMC, which explains the reason why the system efficiency of EMC is better than TMC.

\section{Numerical Simulations}
In this section, we implement our TMC and EMC algorithms, evaluate their performances, and verify whether they satisfy our design rationale separately.

\subsection{Simulation Setup}
To simulate our TMC and EMC, we consider a smart area $\mathbb{B}=(M,\mathbb{V},\mathbb{C})$ with $1000\times 1000$ $km^2$. There are $m$ CSs and $n$ EVs distributed uniformly in this area, where we default by $n=10\cdot m$ unless otherwise specified. For each CS $C_j\in\mathbb{C}$, its number of charging piles $k_j$ is generated from $\{1,2,\cdots,10\}$ randomly with probability $1/10$. The cost $c_j$ of CS $C_j$ is generated according to a uniform distribution within $(0,1]$. Similarly, for each EV $V_i\in\mathbb{V}$, its charging amount $r_i$ is sampled from a truncated normal distribution with mean $50$ and variance $1$ in interval $(0,100]$. To quantify its valuation $v_i^j$ to CS $C_j$, we defined the distance $d_i^j$ between $V_i$ and $C_j$ according to their coordinates, that is
\begin{equation}
	d_i^j=\sqrt{(x_i-x_j)^2+(y_i-y_j)^2}
\end{equation}
where $(x_i,y_i)$ and $(x_j,y_j)$ are the coordinates of $V_i$ and $C_j$. We assume the valuation of an EV to a CS is related to their distance. Thus, the larger $d_i^j$ is, the lower $v_i^j$ is. The maximum distance between two entities in this area is $1000\sqrt{2}$, thus we assume that $v_i^j=1-d_i^j/(1000\sqrt{2})$.

\subsection{Simulation Results and Analysis}

\begin{figure}[!t]
	\centering
	\subfigure[TMC]{
		\includegraphics[width=2.5in]{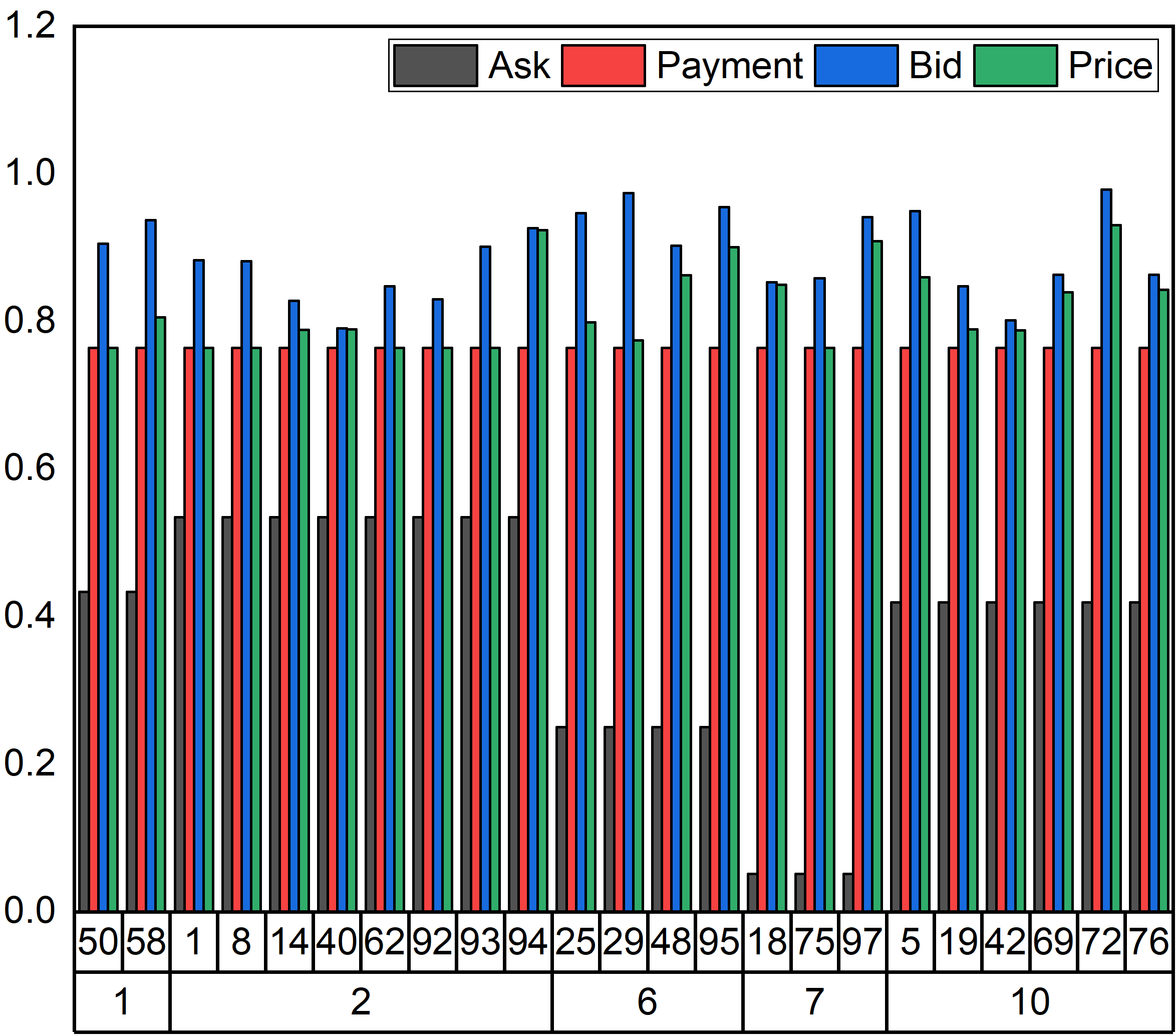}
	}%
	
	\subfigure[EMC]{
		\includegraphics[width=2.5in]{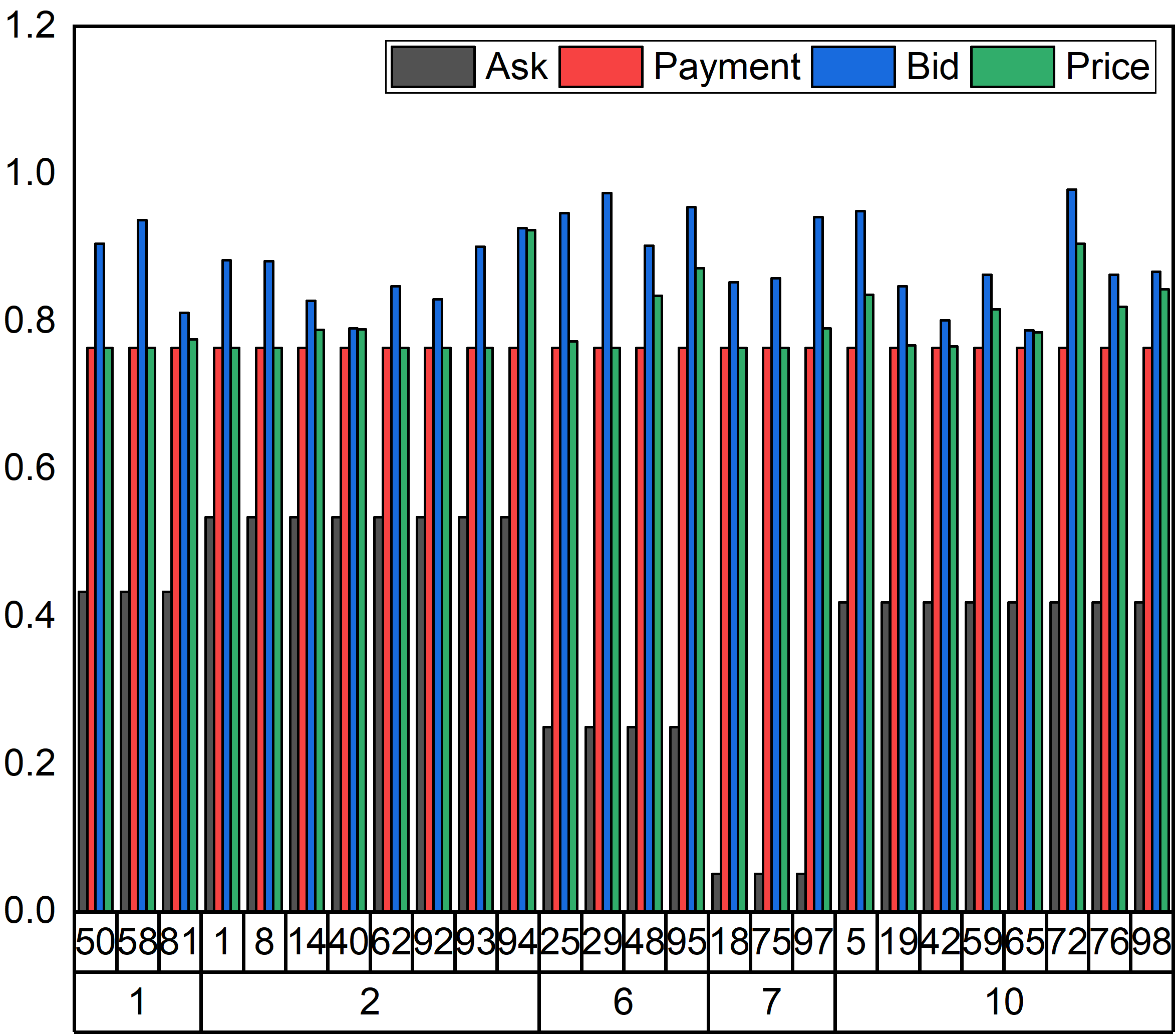}
	}%
	\centering
	\caption{The assignment results and individual rationality obtained by our TMC and EMC.}
	\label{fig4}
\end{figure}

To evaluate individual rationality, budget balance, and truthfulness, we consider a smart area with $m=10$ CSs and $n=100$ EVs. They can be denoted by $\mathbb{C}=\{C_1,C_2,\cdots,C_{10}\}$ and $\mathbb{V}=\{V_1,V_2,\cdots,V_{100}\}$. The median ask is $a_{j_\phi}=0.764$ and there are five winning sellers, thus $\mathbb{C}_w=\{C_1,C_2,C_6,C_7,C_{10}\}$. Moreover, the corresponding number of charging piles of CSs in $\mathbb{C}_w$ is given by $\{r_1:3,r_2:8,r_6:4,r_7:3,r_{10}:8\}$.

\begin{figure}[!t]
	\centering
	\subfigure[Buyer $V_{50}\in\mathbb{V}_w$]{
		\includegraphics[width=0.48\linewidth]{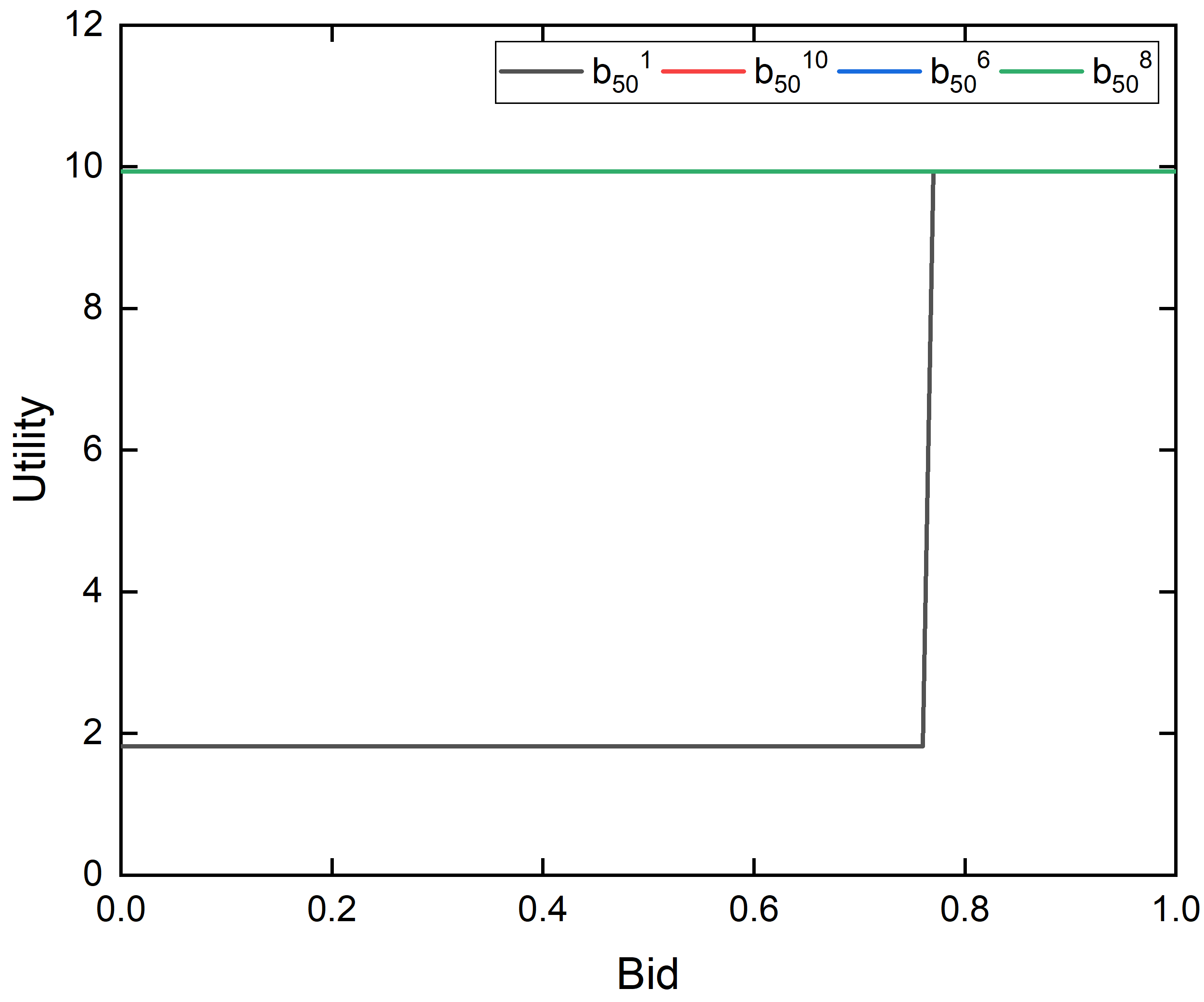}
	}%
	\subfigure[Seller $C_{1}\in\mathbb{C}_w$]{
		\includegraphics[width=0.48\linewidth]{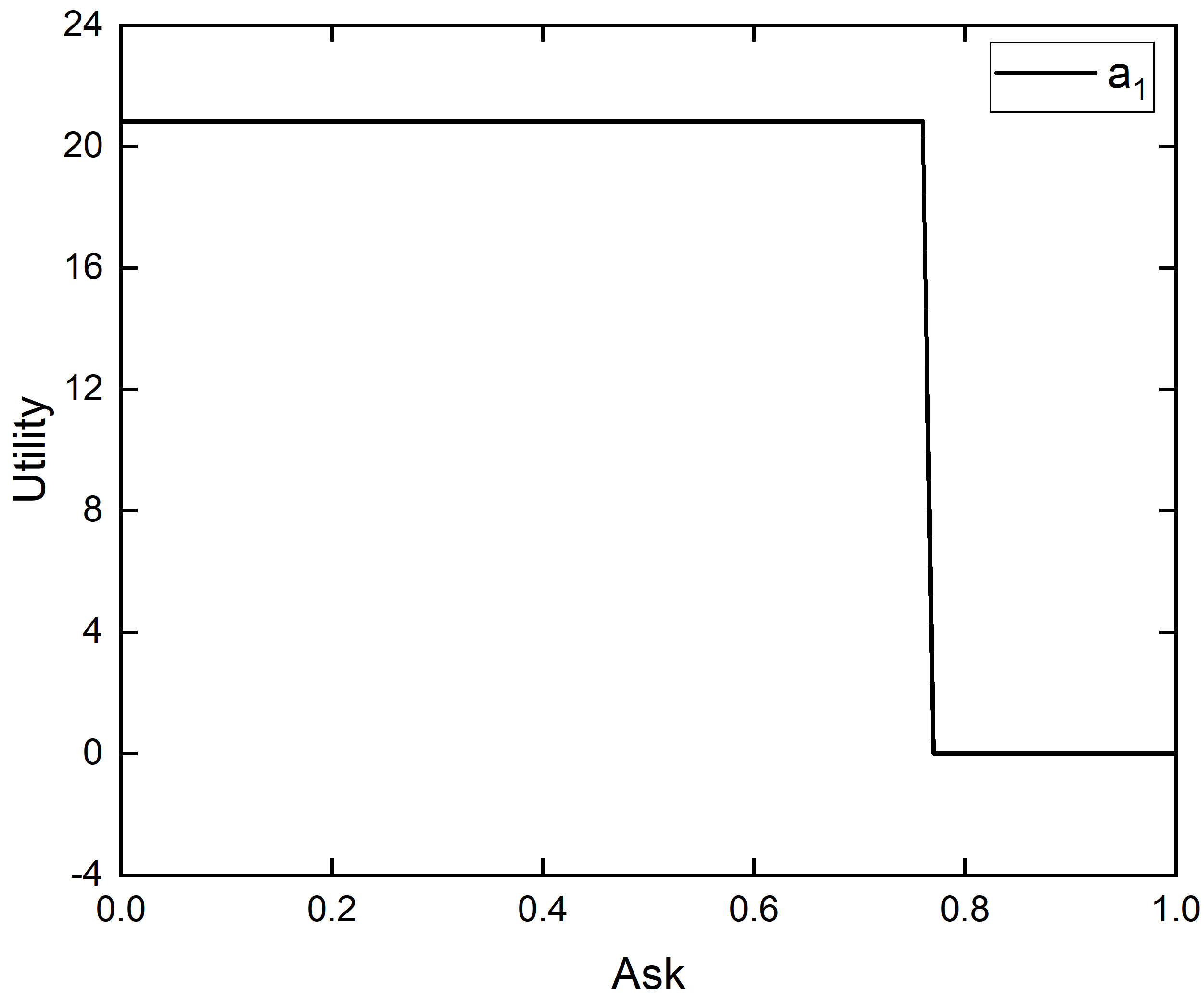}
	}%
	
	\subfigure[Buyer $V_{86}\notin\mathbb{V}_w$]{
		\includegraphics[width=0.48\linewidth]{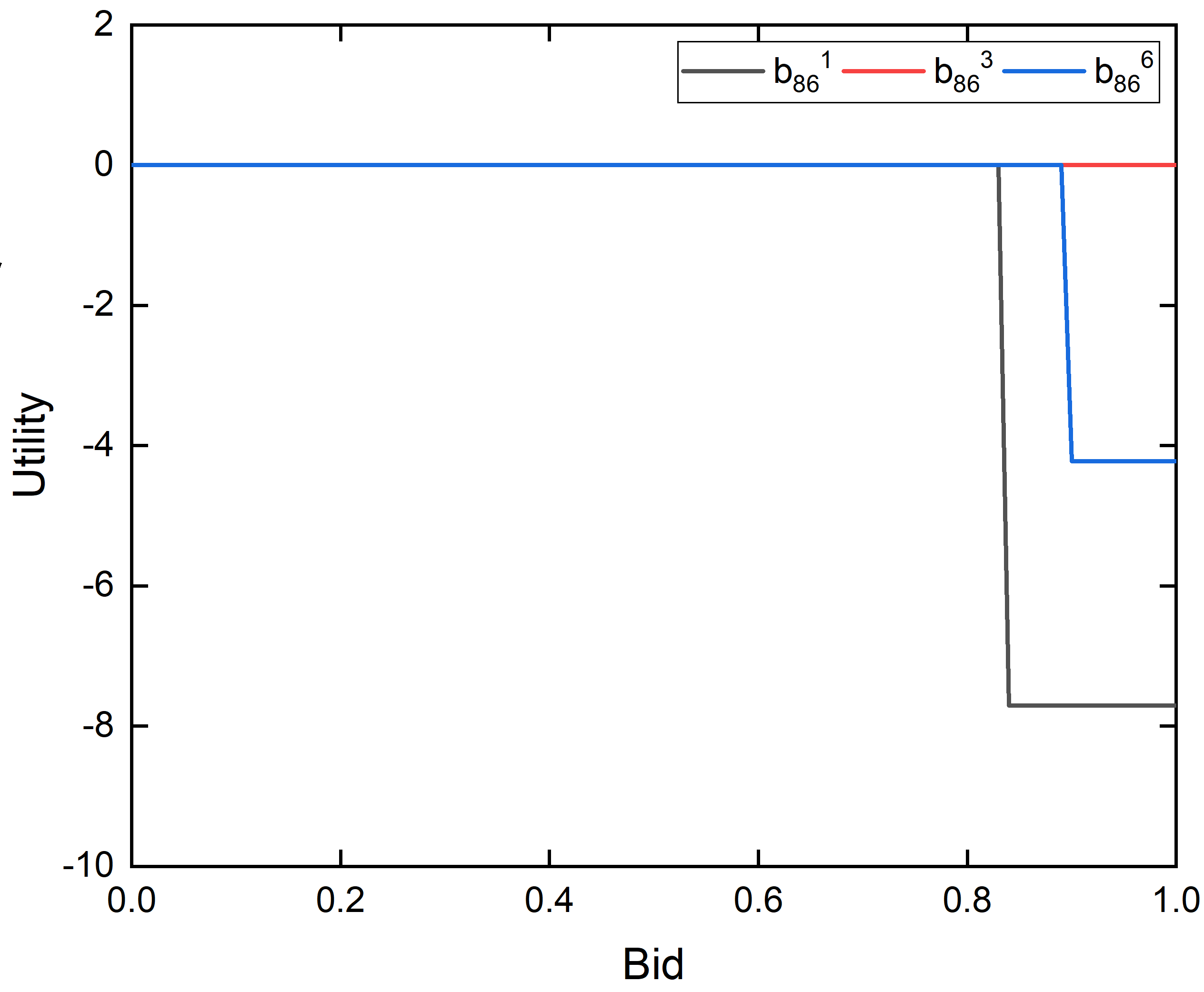}
	}%
	\subfigure[Seller $C_{3}\notin\mathbb{C}_w$]{
		\includegraphics[width=0.48\linewidth]{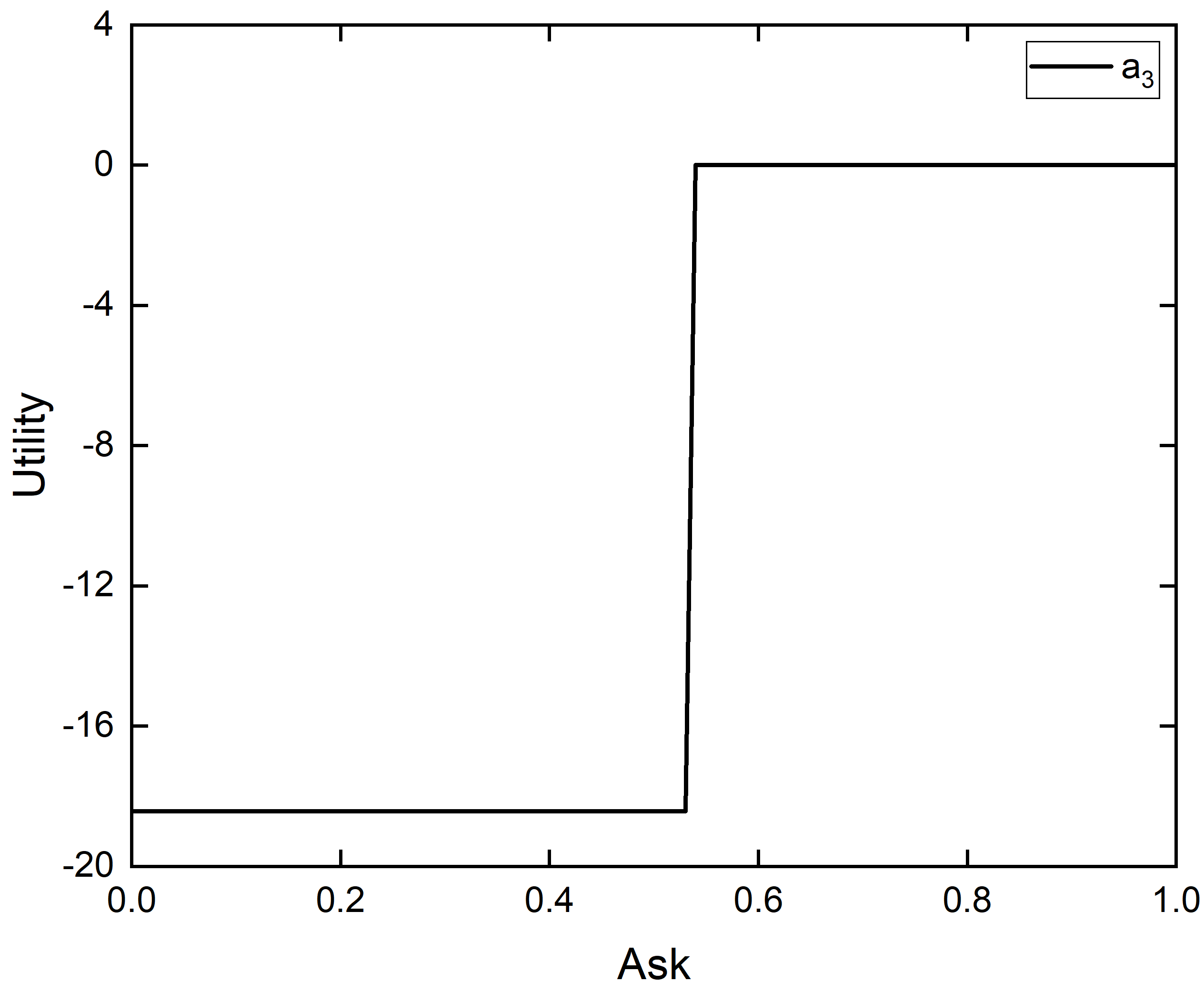}
	}%
	\centering
	\caption{The truthfulness of buyers and sellers in TMC.}
	\label{fig5}
\end{figure}

\begin{figure}[!t]
	\centering
	\subfigure[Buyer $V_{50}\in\mathbb{V}_w$]{
		\includegraphics[width=0.48\linewidth]{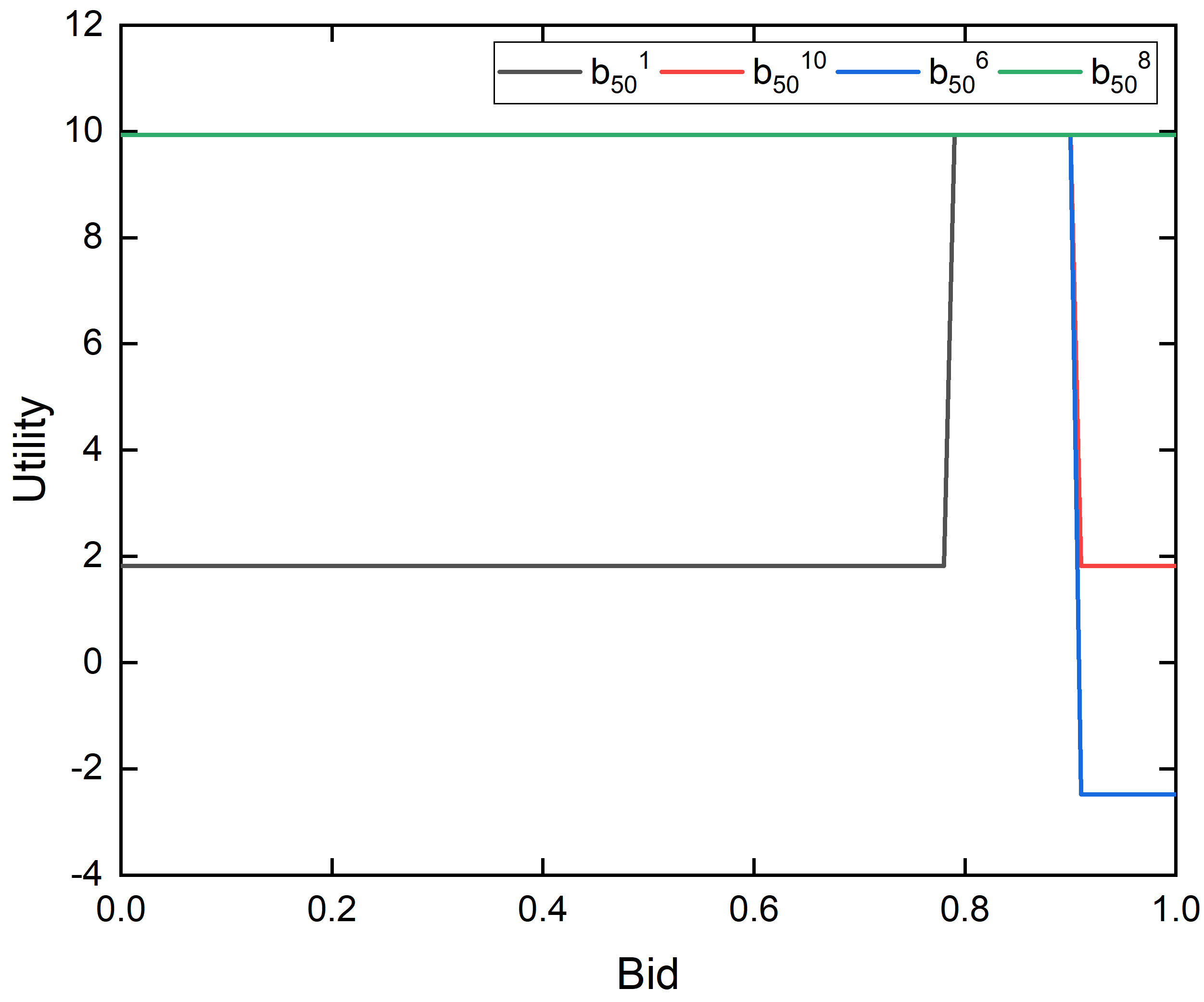}
	}%
	\subfigure[Seller $C_{1}\in\mathbb{C}_w$]{
		\includegraphics[width=0.48\linewidth]{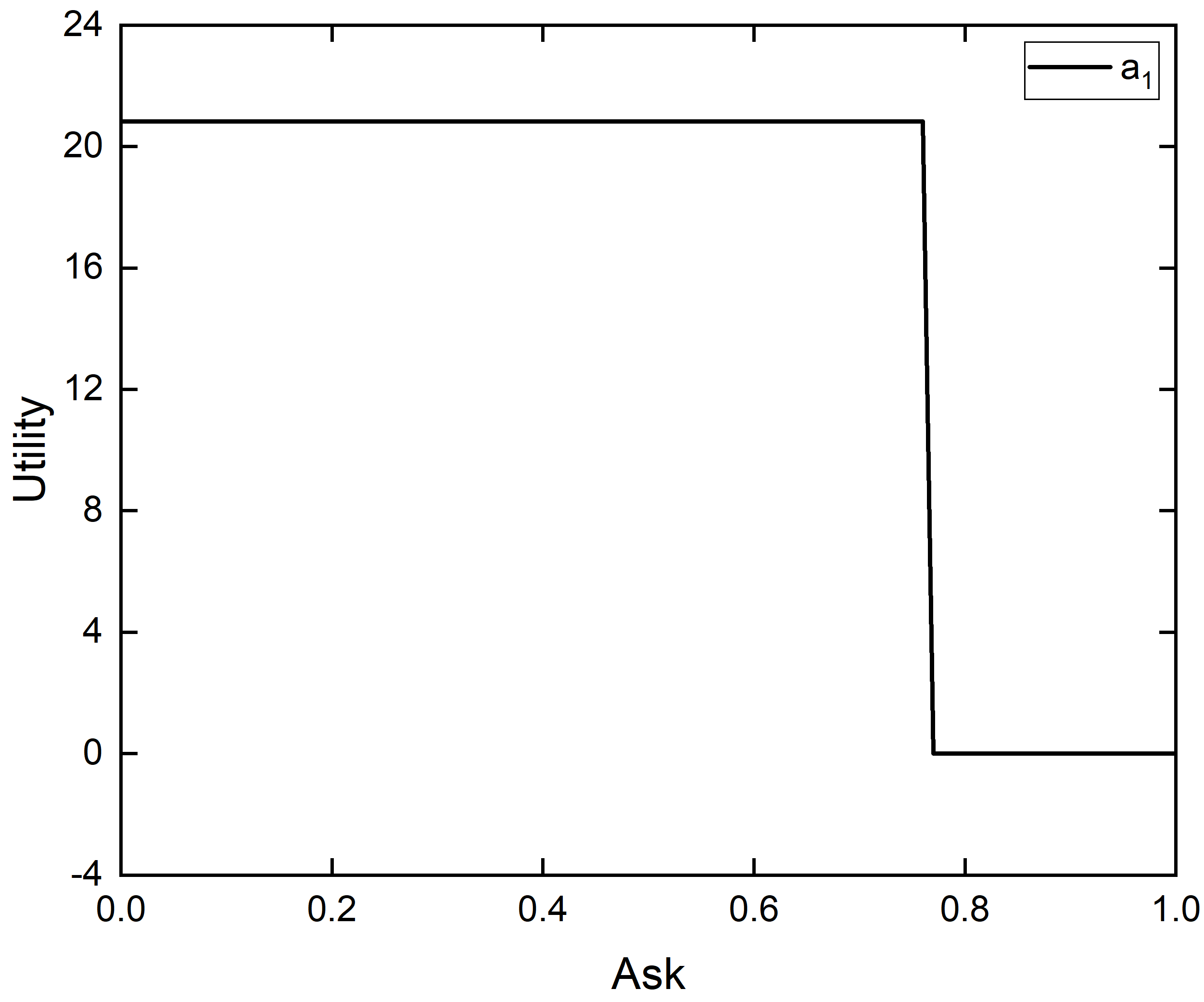}
	}%
	
	\subfigure[Buyer $V_{86}\notin\mathbb{V}_w$]{
		\includegraphics[width=0.48\linewidth]{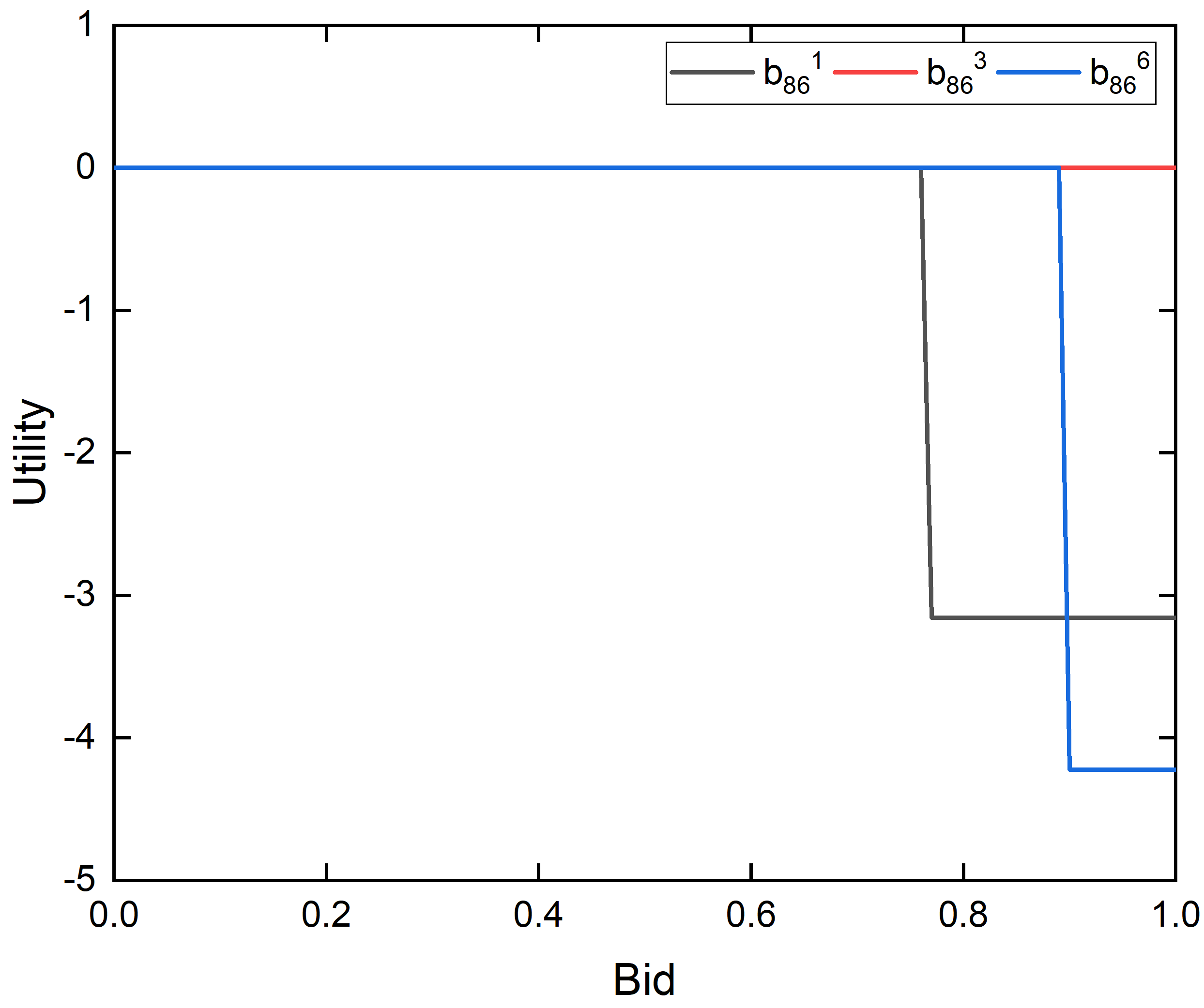}
	}%
	\subfigure[Seller $C_{3}\notin\mathbb{C}_w$]{
		\includegraphics[width=0.48\linewidth]{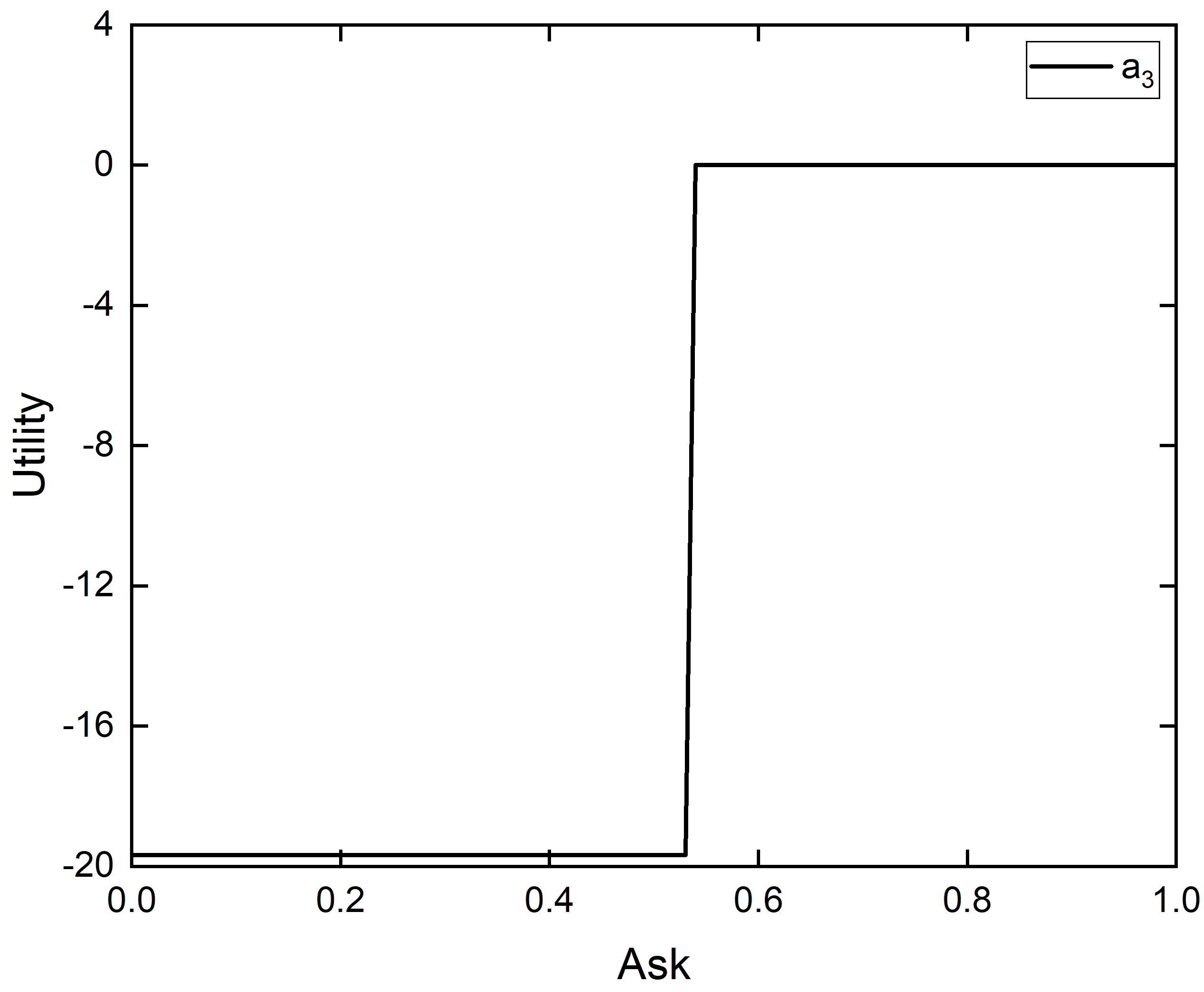}
	}%
	\centering
	\caption{The truthfulness of buyers and sellers in EMC.}
	\label{fig6}
\end{figure}

\textbf{Individual Rationality: }Fig. \ref{fig4} shows the assignment results and individual rationality obtained by TMC and EMC. The first line from the bottom is sellers (CSs) and the second line from the bottom is buyers (EVs). Take Fig. \ref{fig4} (a) as an example, for the seller $C_1$, there are two buyers, $V_{50}$ and $V_{58}$, assigned to it in TMC. For the mapping $\sigma(50)=1$, the payment rewarded to $C_1$ (red column) is more than the ask of $C_1$ (grey column) and the price charged to $V_{50}$ (green column) is less than the bid of $V_{50}$ (blue column). Then, for any mappings from $\mathbb{V}_w$ to $\mathbb{C}_w$ in TMC and EMC, the price charged to the winning buyer is not more than its bid and the payment rewarded to the winning seller is not less than its ask, thus individual rationality is held.

\textbf{Budget Balance: }According to the charged price and rewarded payment shown as Fig. \ref{fig4}, the total price charged to all winning buyers is not less than the total payment rewarded to all winning sellers. Furthermore, the number of buyers $|A_j|$ assigned to each seller $C_j\in\mathbb{C}_w$ is not more than the number of charging piles $k_j$, namely we have $|A_j|\leq k_j$. Thereby the budget balance is held in both TMC and EMC.

\begin{figure}[!t]
	\centering
	\includegraphics[width=2.5in]{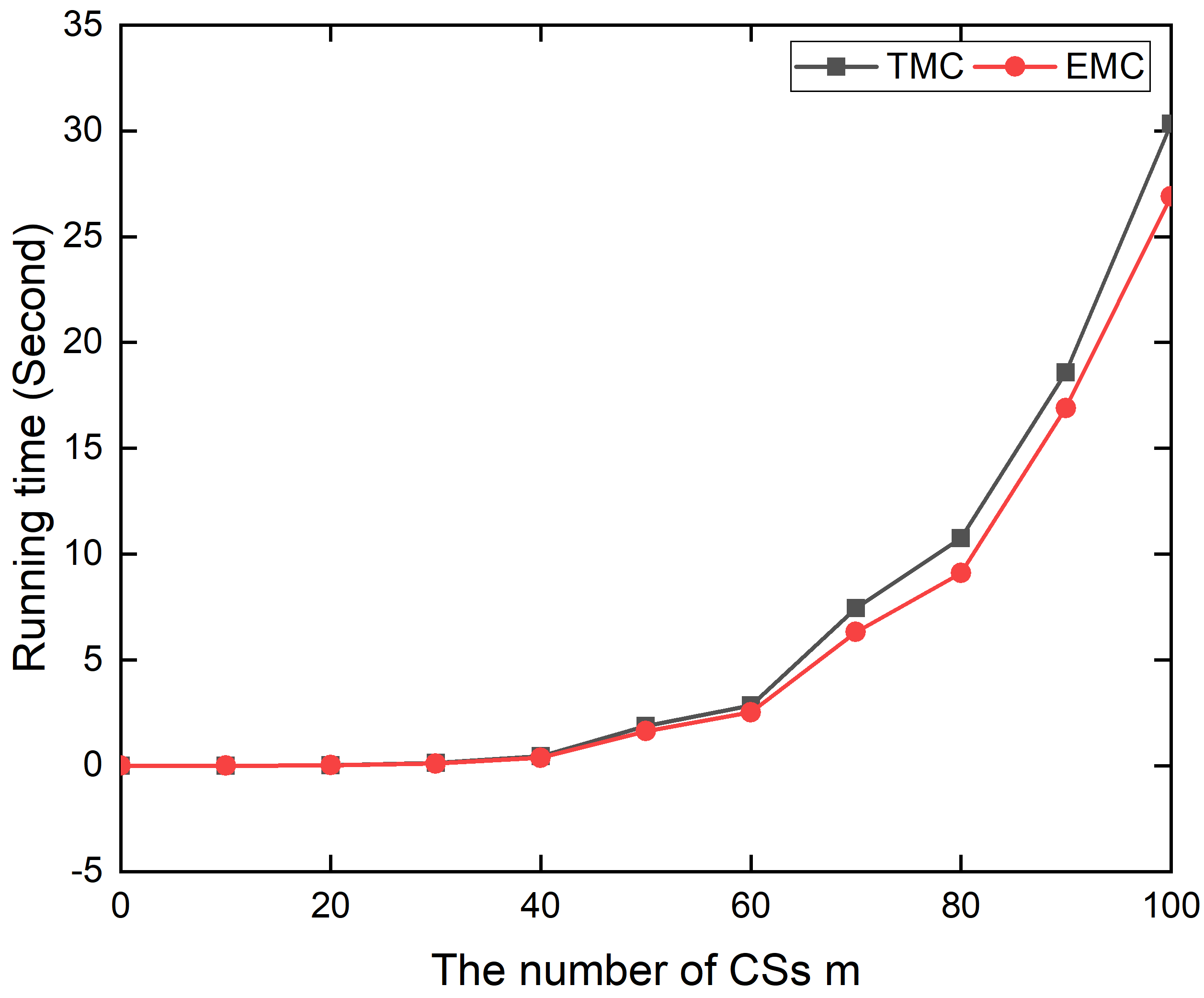}
	\centering
	\caption{The running time varies with the increasing number of CSs in TMC and EMC.}
	\label{fig7}
\end{figure}
\begin{figure}[!t]
	\centering
	\includegraphics[width=2.5in]{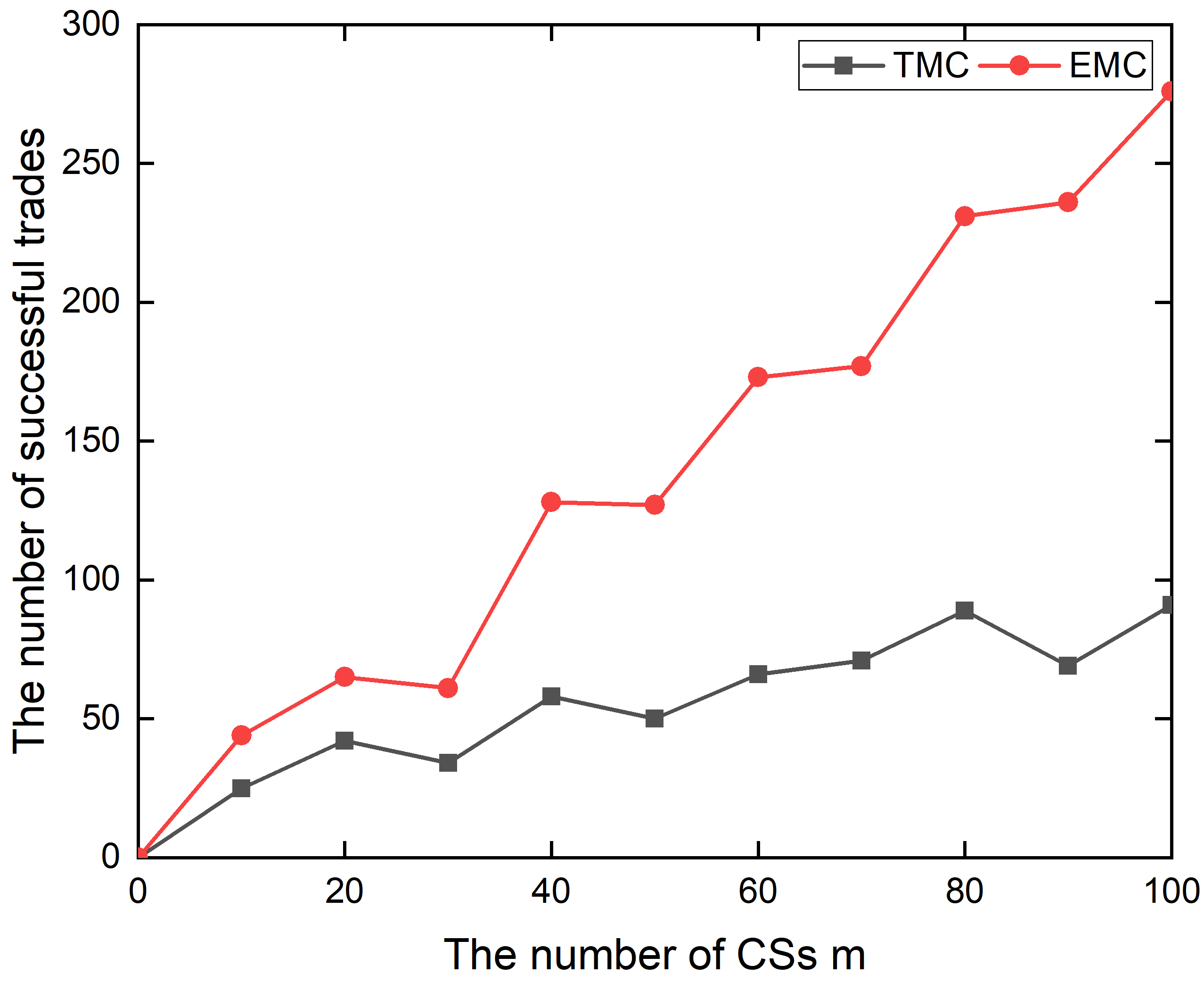}
	\centering
	\caption{The number of successful trades (winning buyers) varies with the increasing number of CSs in TMC and EMC.}
	\label{fig8}
\end{figure}

\textbf{Truthfulness: }We select a winning buyer $V_{50}\in\mathbb{V}_w$, a losing buyer $V_{86}\notin\mathbb{V}_w$, a winning seller $C_1\in\mathbb{C}_w$, and a losing seller $C_3\notin\mathbb{C}_w$ as the representatives to evaluate the truthfulness of buyers and sellers in our TMC and EMC. Fig. \ref{fig5} and Fig. \ref{fig6} show the truthfulness of buyers and sellers in TMC and EMC. Let us look at Fig. \ref{fig5} in TMC first. For the winning buyer $V_{50}$, $\sigma(50)=1$, it can get the maximum utility $\hat{u}_{50}=9.932$ when giving the truthful bid $b_{50}^1=v_{50}^1=0.905$. Here, we have $\mathbb{I}_{50}=\{C_1,C_{10}\}$ and its utility cannot be improved when changing the bids to the seller in $\mathbb{I}_{50}$. If the bid $b_{50}^1<0.77$, its utility will decrease to $1.819$ since the $V_{50}$ will not be selected in $\mathbb{H}_1$ and then be assigned to seller $C_{10}$. Besides, by changing the bids to sellers that are not in $\mathbb{V}_c$ $(C_8)$ or in $\mathbb{V}_c$ but not in $\mathbb{I}_{50}$ $(C_6)$, its utility cannot be improved as well. For the winning seller $C_1$, it can get the maximum utility $\bar{u}_1$ when giving the truthful ask $a_1=c_1=0.434$, which cannot be improved by changing its ask. For the losing buyer $V_{86}$, the utility $\hat{u}_{86}$ is impossible to be more than zero when bidding untruthfully. Its utility will be negative if increasing the bids to sellers in $\mathbb{V}_c$ $(C_3$ and $C_6)$. For the losing seller $C_{3}$, it achieves zero utility when giving the truthful ask $a_3=c_3=0.896$, which will be negative if decreasing its ask. Next, let us look at Fig. \ref{fig6} in EMC. We have the same observations in sub-figures shown as (b), (c), and (d). For the winning buyer $V_{50}$, $\sigma(50)=1$, it has a little different from that in TMC. If increasing the bids to other sellers in $\mathbb{V}_c$ $(C_6$ and $C_{10})$, its utility will decrease, even be negative. This is because the $V_{50}$ will be assigned to $C_6$ or $C_{10}$ instead of $C_1$ since total bids have been varied.

To evaluate the computational efficiency and system efficiency, we consider a smart area whose number of CSs $m$ ranges from $0$ to $200$. The parameters are sampled according to the rules described in the simulation setup.

\textbf{Computational Efficiency: }Fig. \ref{fig7} shows the running time comparison between TMC and EMC. We default by $n=10\cdot m$, thereby the time complexity $O(nm\log(nm))$ can be considered as $10\cdot m^2$ approximately. The trends shown in Fig. \ref{fig7} are in line with our expectations and they are computationally efficient. Besides, we can observe that the running time of EMC is slightly lower than that of TMC since there are more entities eliminated in advance.

\textbf{System Efficiency: }Here, the system efficiency can be characterized by the number of successful trades between buyers and sellers, which is equal to the number of winning buyers in $\mathbb{V}_w$ because each winning buyer will be assigned to a winning seller and then begin to trade. Fig. \ref{fig8} shows the system efficiency comparison between TMC and EMC. The system efficiency is not monotone since we sample the parameters used in this simulation at each number of sellers independently. Shown as Fig. \ref{fig8}, we can observe that the system efficiency in EMC is apparently better than that in TMC, which implies our proposed EMC is an effective approach to improve system efficiency even though it does not guarantee the truthfulness of buyers in some extreme cases. Next, the gap between TMC and EMC increases gradually as the number of sellers increases. This is because the buyer who bids higher can take over more tentative sets of candidate sellers in Algorithm \ref{a3}. However, it can be assigned to only one of these tentative sets, which causes a lot of waste and enlarge the gap.

\subsection{Further Discussion}
According to Lemma 7, we have known that truthfulness is not held for buyers in some extreme cases. A buyer cannot predict an untruthful bid to improve its utility in a deterministic manner because it does not know the bidding strategies of other buyers. For the buyers, it is very risky and difficult to increase their utilities by changing their bids. Our simulation result, shown as Fig. \ref{fig6}, also proves this point that EMC satisfies the truthfulness to some extent. Shown as Fig. \ref{fig7} and Fig. \ref{fig8}, EMC has a lower running time and a much better system efficiency than TMC. Therefore, we prefer to use our EMC instead of TMC in practical applications.

\section{Conclusion}
In this paper, a charging scheduling system based on blockchain technology and a constrained multi-item double auction model has been designed and implemented. To achieve privacy protection and scalability, we gave a lightweight charging scheduling framework based on asymmetric encryption and DAG-based blockchain. To incentivize EVs and CSs to participate in the market, we considered a constrained multi-item double auction model and designed two algorithms, TMC and EMC, that attempt to assign EVs (buyers) in this area to be charged in CSs (sellers). Both algorithms are feasible, which ensures individual rationality, budget balance, truthfulness, and computational efficiency. Here, EMC can get a better system efficiency than TMC, but it weakens the truthfulness of buyers to some extent. Finally, the results of numerical simulations indicated that our model is robust and theoretical analysis is correct.

\section*{Acknowledgment}

This work was supported in part by the National Natural Science Foundation of China (NSFC) under Grant No. 62202055 and No. 62202016, the Start-up Fund from Beijing Normal University under Grant No. 310432104, the Start-up Fund from BNU-HKBU United International College under Grant No. UICR0700018-22, the Project of Young Innovative Talents of Guangdong Education Department under Grant No. 2022KQNCX102, and the National Science Foundation (NSF) under Grant No. 1907472 and No. 1822985.

\ifCLASSOPTIONcaptionsoff
  \newpage
\fi

\bibliographystyle{IEEEtran}
\bibliography{references}

\begin{IEEEbiography}[{\includegraphics[width=1in,height=1.25in,clip,keepaspectratio]{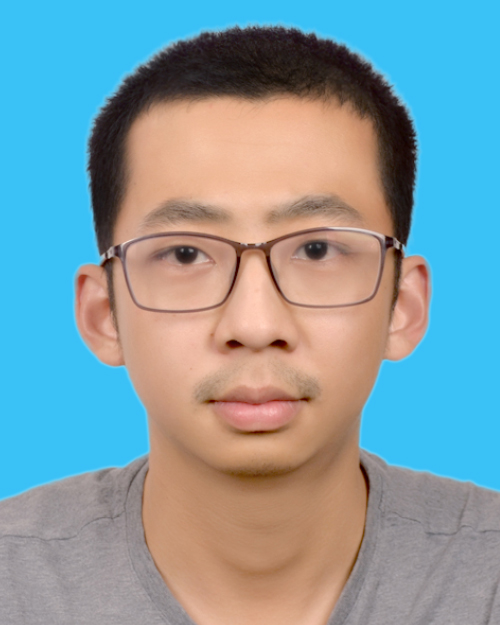}}]{Jianxiong Guo}
	received his Ph.D. degree from the Department of Computer Science, University of Texas at Dallas, Richardson, TX, USA, in 2021, and his B.E. degree from the School of Chemistry and Chemical Engineering, South China University of Technology, Guangzhou, China, in 2015. He is currently an Assistant Professor with the Advanced Institute of Natural Sciences, Beijing Normal University, and also with the Guangdong Key Lab of AI and Multi-Modal Data Processing, BNU-HKBU United International College, Zhuhai, China. He is a member of IEEE/ACM/CCF. He has published more than 40 peer-reviewed papers and been the reviewer for many famous international journals/conferences. His research interests include social networks, wireless sensor networks, combinatorial optimization, and machine learning.
\end{IEEEbiography}

\begin{IEEEbiography}[{\includegraphics[width=1in,height=1.25in,clip,keepaspectratio]{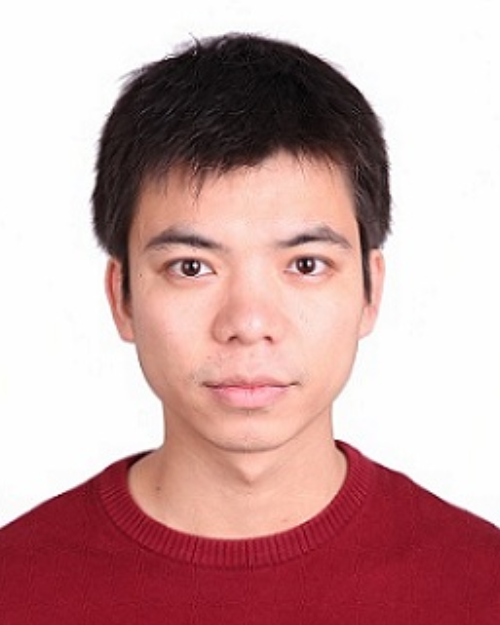}}]{Xingjian Ding}
	received his B.E. degree in electronic information engineering from Sichuan University in 2012 and M.S. degree in software engineering from Beijing Forestry University in 2017. He obtained his Ph.D. degree from the School of Information, Renmin University of China in 2021. He is currently an assistant professor at the School of Software Engineering, Beijing University of Technology. His research interests include wireless rechargeable sensor networks, approximation algorithms design and analysis, and blockchain.
\end{IEEEbiography}

\begin{IEEEbiography}[{\includegraphics[width=1in,height=1.25in,clip,keepaspectratio]{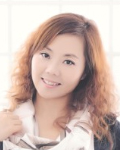}}]{Weili Wu}
	received the Ph.D. and M.S. degrees from the Department of Computer Science, University of Minnesota, Minneapolis, MN, USA, in 2002 and 1998, respectively. She is currently a Full Professor with the Department of Computer Science, The University of Texas at Dallas, Richardson, TX, USA. Her research mainly deals in the general research area of data communication and data management. Her research focuses on the design and analysis of algorithms for optimization problems that occur in wireless networking environments and various database systems.
\end{IEEEbiography}

\begin{IEEEbiography}[{\includegraphics[width=1in,height=1.25in,clip,keepaspectratio]{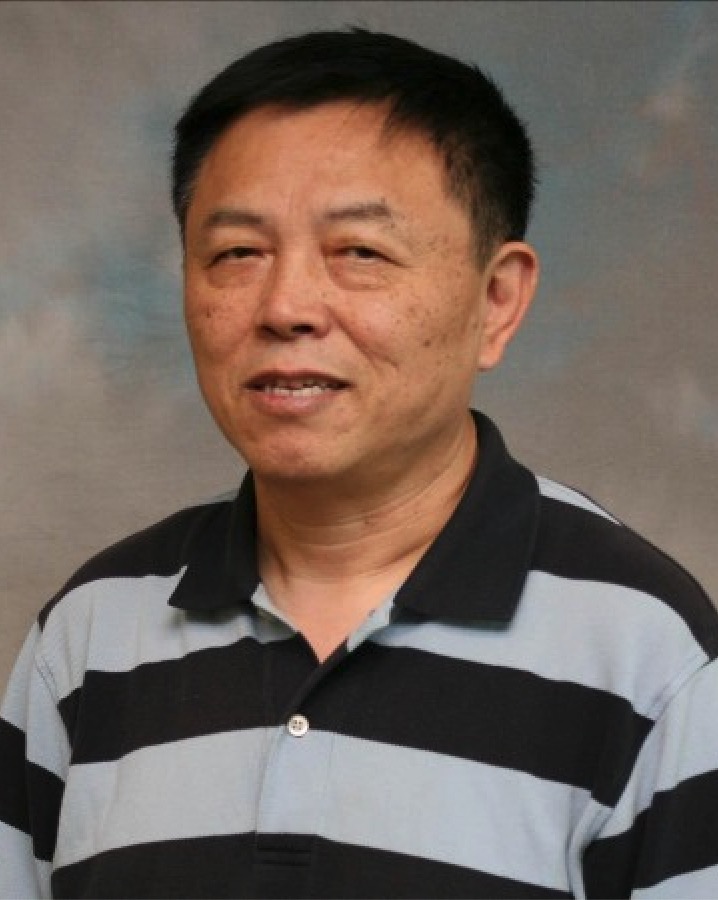}}]{Ding-Zhu Du}
	received the M.S. degree from the Chinese Academy of Sciences, Beijing, China, in 1982, and the Ph.D. degree from the University of California at Santa Barbara, Santa Barbara, CA, USA, in 1985, under the supervision of Prof. R. V. Book. Before settling at The University of Texas at Dallas, Richardson, TX, USA, he was a Professor with the Department of Computer Science and Engineering, University of Minnesota, Minneapolis, MN, USA. He was with the Mathematical Sciences Research Institute, Berkeley, CA, USA, for one year, with the Department of Mathematics, Massachusetts Institute of Technology, Cambridge, MA, USA, for one year, and with the Department of Computer Science, Princeton University, Princeton, NJ, USA, for one and a half years. Dr. Du is the Editor-in-Chief of the Journal of Combinatorial Optimization and is also on the editorial boards for several other journals.
\end{IEEEbiography}

\end{document}